\documentclass{lmcs} 


\usepackage{hyperref}
\theoremstyle{plain} 
\usepackage{amsmath}
\usepackage{amssymb}
\usepackage{thm-restate}
\usepackage{csquotes}
\usepackage[shortlabels]{enumitem}
\usepackage{xcolor}

\newcommand{\Formula}[2]{{\mathfrak{Form}^{#1}_{#2}}}

\def\bp{\bar{P}}
\def\bA{\bar{A}}
\def\bq{\bar{Q}}
\def\bx{\bar{x}}
\def\bX{\bar{X}}
\def\by{\bar{y}}
\def\bY{\bar{Y}}

\def\bE{Succ_0,\dots,Succ_{k-1}}

\def\Sigmain{\Sigma_{in}}
\def\Sigmaout{\Sigma_{out}}

\def\As{\mathcal{AS}}
\def\AStore{\mathcal{S}}
\def\cA{\mathcal{A}}
\def\self{\mathit{self}}
\def\forward{\mathit{forward}}
\def\backward{\mathit{backward}}

\def\om{\omega}

\def\nat{{\mathbb{N}}}

\def\mM{\mathcal M}
\def\Nat{\mathbb{N}}

\def\MSO{\mathsf{MSO}}

\begin{document}

\title{Synthesis of Infinite State Systems}

\author{Ohad Drucker}
\author{Alexander Rabinovich}[a]

\address[a]{Tel Aviv University, Israel}

\begin{abstract}
  \noindent The classical Church synthesis problem, solved by B\"uchi and Landweber \cite{BL69}, treats the synthesis of finite state systems. The synthesis of infinite state systems, on the other hand, has only been investigated few times since then, with no complete or systematic solution.

We present a systematic study of the synthesis of infinite state systems. The main step involves the synthesis of $\MSO$-definable parity games, which is, finding $\MSO$-definable uniform memoryless winning strategies for these games.
\end{abstract}

\maketitle

\section{Introduction}

Let Spec be a specification language and Pr be an implementation language. The following problem is known as the Synthesis problem:

\medskip

\noindent $\mbox{}$\hspace{0.1cm} \framebox {
\begin{minipage}{0.95\hsize}
\noindent \hspace{2.0cm}
{\bf  Synthesis Problem}\vskip4 pt
\noindent{\em Input:} A specification $S(I,O)\in Spec$.\\

\noindent{\em Question:} Is there a program $P\in Pr$ which implements it, i.e., $\forall I\ S(I,P(I))$ holds.
\end{minipage}
}
\medskip

Our aim is a systematic study of the generalized Church synthesis problem.

We first discuss the \textbf{Church synthesis problem over $\omega=(\Nat,<)$}, and then explain its generalizations.

\subsection{The Church synthesis problem over $\omega$}

A synthesis problem is parameterized by a class of specifications and a class of implementations.

The \textbf{specification language} for the Church synthesis problem over $\omega$ is Monadic Second-Order Logic of Order ($\MSO$).
In an $\MSO$ formula $\varphi(X,Y)$, free variables   $X$ and $Y$  range over  monadic predicates over $\Nat$.
Since each such monadic predicate  can be identified with its characteristic $\omega$-string, $\varphi$ defines a binary relation on
  $\omega$-strings. In particular  a specification $S(I,O)$ relates  input $\omega$-strings $I$ and output $\omega$-strings $O$.

The \textbf{implementations} should then be functions from  $\omega$-strings to $\omega$-strings, which we will call operators. However, one would want to consider operators that may be implemented by machines. When a machine computes an operator,
at every moment $t\in \mathbb{N}$ it reads an input symbol $X(t)\in \{0, 1\}$, updates its internal state, and produces an output
symbol $Y (t) \in \{0, 1\}$. Hence, the output $Y(t)$ produced at $t$ depends only on inputs
symbols $X(0), X(1),\dots, X(t)$. Such operators are called \textbf{causal} operators, and these are the implementations of the Church synthesis problem over $\omega$:

\medskip

\noindent $\mbox{}$\hspace{0.1cm} \framebox {
\begin{minipage}{0.95\hsize}
\noindent \hspace{2.0cm}
{\bf  Church Synthesis problem over $\omega$}\vskip4 pt
\noindent{\em Input:} An $\MSO$ formula $\Psi(X,Y)$.\\

\noindent{\em Question:} Is there a causal operator F such that
$ \forall X\ \Psi(X,F(X))$ holds over $\Nat$?
\end{minipage}
}
\medskip

\subsection{B\"{u}chi and Landweber's solution}

B\"{u}chi and Landweber \cite{BL69} proved that the Church synthesis problem is
computable. Their main theorem can be stated as: 
\begin{thm}
\label{BLthm}
Let $\Psi(X,Y)$ be an $\MSO$ formula.
\begin{enumerate}
    \item (Decidability) It is decidable whether there is a causal operator $F$ such that $\mathbb{N}\models \forall X\ \Psi(X,F(X))$.
    \item (Definability) If yes, then there is an $\MSO$ formula that defines a causal operator which implements $\Psi$.
    \item (Computability) There is an algorithm that constructs, for each implementable $\MSO$ formula $\Psi(X,Y)$, an $\MSO$ formula $\varphi(X,Y)$ which defines a causal operator that implements $\Psi$.
\end{enumerate}
\end{thm}
\smallskip

Note that a causal operator is $\MSO$-definable if and only if it is computable by a finite state automaton.

The proof of Theorem \ref{BLthm} is based on the fundamental connections between Logic, Games and Automata. The proof has four main steps:

\begin{description}
\item[From logic to automata] Construct a deterministic finite state automaton equivalent to a given logical formula.
\item[From automata to games] Convert the  automaton  into a (finite) game arena for a two person perfect information infinite
game.
\item[Solve the game] Find who wins the game and find a winning strategy.
\item[From winning strategies to input-output automata]
Given a winning strategy construct an automaton that implements the logical specification.
\end{description}
\smallskip

It is worth mentioning that although the synthesis problem over $\omega$ is traditionally attributed to Church, in \cite{Church57} Church does not explicitly restrict the discussion to finite state
systems. He has a vague and general formulation of “logistic
systems” and “circuits” and in fact discusses infinite state systems:
“Given a requirement which a circuit is to satisfy, we may suppose the
requirement expressed in some suitable logistic system which is an
extension of restricted recursive arithmetic. The synthesis problem is
then to find recursion equivalences representing a circuit that satisfies
the given requirement (or alternatively, to determine that there is no
such circuit).”

However, following B\"uchi and Landweber \cite{BL69} the community narrowed the
view of Church’s Problem to the finite-state case. B\"uchi and Landweber's game theoretical
methods for solving the synthesis problem initiated a line of research in the
algorithmic theory of infinite games that has been extensively
developed in recent years. Enormous
technical and conceptual progress has been achieved in developing game
methods,  and a game-theoretic
framework for models of interactive computation has been established.
Game-theoretic methods have also been extended to study synthesis
problems for real-time and hybrid
automata \cite{AFHMS03,BHPR07,BouyerFLMS08,AP18,CL20}. At the same time, comparatively little attention has been paid to the original
motivation of the Church problem, with a relatively small number of  works directly addressing synthesis of infinite state systems. In that sense, our paper takes the discussion back to its original context. 

\subsection{Review of synthesis results for infinite state systems}

Before we introduce our results, we review the known results involving synthesis of infinite state systems.

In \cite{Wal01} Walukiewicz considered arenas generated by \textbf{pushdown automata}.
A configuration of a pushdown automaton contains the content of a stack and a control state.
The set of such configurations is infinite and a pushdown automaton defines in a natural way a transition relation (graph)
on its configurations.
 Walukiewicz \cite{Wal01}
  proved: there is an algorithm that decides who is the winner of the parity game over the arena described by a given pushdown automaton; moreover, the winner
has a winning strategy which can be implemented by a pushdown transducer.

These results were generalized in \cite{CarayolHMOS08,Walukiewicz16} to \textbf{higher-order pushdown automata}.
Carayol et al.  proved: there is an algorithm that decides who is the winner of the parity game over the arena described by a given higher-order pushdown automaton; moreover, the winner
has a winning strategy which can be implemented by a higher-order pushdown transducer.
Recently, it was shown that a similar theorem holds for collapsible pushdown automata \cite{Bro21,carayol2010choice}.

The arena for the above games is infinite, however the out degree of every node is finite
and these games can be considered as games over an infinite arena with  a \textbf{finite alphabet}.

One case of games over an \textbf{infinite alphabet} that can be found in the literature is the synthesis of games over \textbf{prefix-recognizable graphs} \cite{CachatWaluk}. Cachat showed that the synthesis problem of such games may be efficiently reduced to the synthesis problem of pushdown automata.

The above examples are about synthesis of games over infinite arenas, namely, finding definable winning strategies for such games.
Brütsch and Thomas \cite{BRUTSCHLT17,BrutschT22} explicitly considered a generalization of Church problem in which infinite state systems over infinite alphabets are to be synthesized.
They introduced $\nat$-\textbf{automata} that define  specifications over the alphabet of natural numbers. They also introduced transducers
that can implement such relations.
Their main results show that it is decidable whether the relation defined by an $\nat $-automaton is implementable by a transducer, and if so, there is an algorithm that computes such a transducer.

There is also a line of work on the synthesis of reactive systems over infinite alphabets equipped with equality or a linear order \cite{BhaskarP24, Exibard21, ExibardFK22, KhalimovK19}.
These works are based on the study of constraint sequences, which abstract the behaviour of register automata and allow the reduction of Church games to $\omega$-regular games.

Altogether, these examples are instances of synthesis of infinite state systems that were carefully studied.
However, there are no systematic studies of this topic.

\smallskip

\smallskip

\subsection{Overview of our results}
Our first step in studying synthesis of infinite state systems is synthesis of \textbf{games over infinite arenas}. In order to study algorithmic questions about these games, we need a formalism that provides a finite description of infinite state arenas, and a formalism that provides a finite  description of strategies over infinite arenas.

Examining a concrete example, in the synthesis of $\nat$-automata  considered in \cite{BRUTSCHLT17,BrutschT22}, the resulting games are parity games
\textbf{$\MSO$-definable in $\omega=(\nat ,<)$}. Although any parity game has simple winning strategies, in the sense that they are uniform and memoryless, their descriptive complexity might be high. Thus, we may ask a general question on $\omega$: given an  $\MSO$-definition of a parity game in $\omega$, can we decide who wins the game and describe a memoryless winning strategy? 

We answer this question positively by showing that this synthesis problem is decidable, and by finding a memoryless winning strategy that is definable
in $\omega$. The same positive answer applies for expansions of $\omega$ by parameters. Moreover, an implementing $\nat$-transducer may efficiently be constructed out of that definable strategy.

Another concrete example to examine is the parity game described  by a  pushdown  graph. Such a game is \textbf{$\MSO$-definable} in the \textbf{full binary tree}. The previous question now arises in this wider context: given an $\MSO$-definition of a parity game in the full binary tree, can we decide who wins the game and describe a memoryless winning strategy?

We answer this question positively as well, a result that generalizes the  Walukiewicz  theorem. Our proof can be modified
to obtain a simple proof of the  Walukiewicz theorem (although this simple proof does not give any complexity analysis as the one given by Walukiewicz). Our result also applies to some expansions of the full binary tree by parameters.

More generally, our approach to investigate synthesis of games over infinite arenas is
parameterized by a structure $\mM$ and can be stated as:

\medskip

\noindent $\mbox{}$\hspace{0.1cm} \framebox {
\begin{minipage}{0.95\hsize}
\noindent \hspace{1.4cm}
{\bf  Game Synthesis Problem for $\mM$}\vskip4 pt
\noindent{\em Input:} An $\MSO$ definition of a parity game in $\mM$.\\

\noindent{\em Task:} Decide who wins the game;

  \hspace{0.88cm} Find an $\MSO$ definition of a memoryless winning strategy.
\end{minipage}
}\medskip

For the precise definitions, see subsection \ref{preliminary_games}.

Our main results provide sufficient conditions for the decidability of the game synthesis problem for $\mathcal{M}$ and the existence of a winning  strategy that is $\MSO$-definable in $\mM$.

The notion of \textbf{selection} plays a vital role in these results. Let $\varphi(X_1,\dots X_l)$ be a formula. A formula $\psi(X_1,\dots ,X_l)$ is a \textbf{selector} for $\varphi$ in a model $\mM$ if whenever $\varphi$ is satisfiable in $\mM$, there is a unique tuple that satisfies $\psi$ and it satisfies $\varphi$. We say that $\mathcal{M}$ has the \textbf{selection property} if all satisfiable formulas have a selector. Extensive investigation of the selection property may be found in the literature, such as \cite{LifschesS98,RabinovichS08,RS08,Rabinovich07a}. Selection for iterated pushdown memory was studied in \cite{Fratani}.

In many cases there is a formula that describes the set of memoryless winning strategies for a given definable parity game. If selection fails for this formula, one  cannot describe a single memoryless  winning strategy. However, we show:

\begin{thm}[\textbf{Definable strategies in bounded degree games}]
If $\mathcal{M}$ has the selection property, and $G$ is a parity game $MSO$-definable in $\mathcal{M}$ such that $G$ has bounded degree, then there are MSO-definable in $\mathcal{M}$ uniform memoryless winning strategies for $G$.
\end{thm}

This result is then extended to \textbf{unbounded degree games}, out of which we get:

\begin{thm}
Let $\mathcal{M}$ be an expansion by unary predicates of an ordinal $< \omega^\omega$, or an expansion by unary predicates of the full $k$-ary tree which has the selection property. Then MSO-definable in $\mathcal{M}$ parity games have MSO-definable in $\mathcal{M}$ uniform memoryless winning strategies.
\end{thm}

The next step is applying these results on the \textbf{synthesis of infinite state systems}. In order to properly define the synthesis problem, we need a specification language that describes an input-output relation. A natural possibility is to consider a definable parity game between players input and output whose edges are labelled by some alphabets $\Sigmain, \Sigmaout$ that are not necessarily finite. The transitions are deterministic, which is, a vertex has at most one outgoing edge for each label. However, a single outgoing edge may have many different labels. The input-output relation defined by the game is then the set of pairs $(x,y) \in {\Sigmain}^\omega \times {\Sigmaout}^\omega$ which label a play that is winning for player output.

The generalized Church problem can then be stated as:

\medskip

\noindent $\mbox{}$\hspace{0.1cm} \framebox {
\begin{minipage}{0.95\hsize}
\noindent \hspace{1.2cm}
{\bf  Generalized Church Synthesis Problem for $\mM$}\vskip4 pt
\noindent{\em Input:} An $\MSO$ definition of an edge labelled parity game in $\mM$.\\ $\mbox{} \hspace{1.0cm}$

\noindent{\em Question:} Is there a causal operator that implements the input-output relation defined by the game?

\hspace{1.65cm} If yes, can the relation be implemented by a transducer that is $\MSO$-definable in $\mathcal{M}$?
\end{minipage}
}
\medskip

For the precise definitions, see subsection \ref{GCPInfSect}.

To obtain a solution of this generalized Church problem from the solution of the game synthesis problem, we need to be able to \textbf{choose} one label of the transition specified by a winning strategy \textbf{in a definable way}. When the set of labels is finite, that can easily be done. However, when the set of labels is infinite that is not necessarily possible, and we provide sufficient conditions for that definable choice of labels to exist.

In particular, we prove:
\begin{thm}[\textbf{Synthesis of definable infinite state systems}]
Let $\mathcal{M}$ be an ordinal $< \omega^\omega$, or the full $k$-ary tree. Let $\mathcal{G}$ be an MSO-definable in $\mathcal{M}$ edge labelled parity game over alphabets $\Sigmain$ and $\Sigmaout$ (not necessarily finite).
\begin{enumerate}
\item (Decidability) It is decidable whether there exists a causal operator that implements the relation induced by $\mathcal{G}$.
\item (Definability) If yes, then there is an MSO-definable in $\mathcal{M}$ transducer over $\Sigmain$ and $\Sigmaout$ that implements the relation induced by $\mathcal{G}$.
\end{enumerate}
\end{thm}

\smallskip

The paper is organized as follows.
Section \ref{Prelim} contains preliminaries.
Section \ref{SectionBoundedOutDegree} considers definable games over graphs of bounded out degree.
It provides sufficient conditions for the existence of definable memoryless  winning strategies for these games.
Section \ref{SectionUnbounded} generalises these results to games of unbounded out degree.
Section \ref{GCPSection} considers the generalized Church problem over finite and infinite alphabets. 

\section{Preliminaries}\label{Prelim}
This long preliminary  recalls  standard notations and results about logic, automata and games.
In addition, it covers elements of the composition method  in subsection \ref{sect:comp} and the selection property in subsection \ref{sect:select}.

\subsection{Notations and terminology}
We use $n,m,i,j,k,l$ for natural numbers, $\omega$ for the set of natural numbers, and $\alpha,\beta$ for ordinals.

If $X$ is a set, $2^X$ is the power set of $X$, which is, the set of subsets of $X$. A finite sequence of elements of $X$ is called a word or a string over $X$.
The empty string is denoted by $\epsilon$. The set of words over $X$ is denoted by $X^{<\omega}$ or $X^*$. An $\omega$-word over $X$ is a function $f:\omega \to X$.
If $u,w$ are words over $X$, $u\cdot w$ is their concatenation, and $u^\omega$ is the $\omega$-word formed by infinitely many concatenations of $u$. If $L$ is a set of words, $L^+=\{w_0\cdots w_n\ |\ n\in\omega, \forall i\ w_i \in L\}$, and $L^*=L^+ \cup \{\epsilon\}$.

For a tuple of sets $\bar{A}=(A_0,\dots,A_{n-1})$, $\bar{A}\subseteq X$ means that for every $0\leq i \leq n-1$, $A_i \subseteq X$.

We use the symbol $\simeq$ for isomorphism. If $f:A\to B$ and $A' \subseteq A$, then $f(A'):=\{f(x)\ |\ x \in A'\}$.

If $(A,<)$ is a linearly ordered set, we use standard notation for its intervals: $(a,b)=\{x \in A\ |\ a < x < b\}$, $[a,b):=(a,b)\cup\{a\}$, etc. The interval $(a,\infty)$ is $\{x \in A\ |\ x > a\}$, and $(-\infty,a):=\{x \in A\ |\ x < a\}$.

The term `effectively determines / represents / defines' means that there is an algorithm that determines, or that finds the required representation or definition. Similarly, when an object is `computable', it means that there is an algorithm that finds it.

\subsection{Monadic Second Order Logic}
We use standard notations and terminology about Monadic Second Order logic (see, for example, \cite{ThomasMSO}). In this version the first order variables are eliminated and only second order variables that are interpreted as unary predicates are used.
\begin{defi}[$\MSO$]
\begin{enumerate}
\item (Syntax) Let $\mathcal{L}=\langle R_i\ |\ i \in I \rangle$ be a relational signature, where $R_i$ is a $k_i$-ary relation symbol.
The \textbf{vocabulary} of $\MSO$ over $\mathcal{L}$ is $\mathcal{L}$ together with the \textbf{monadic second order variables} $\{X_i\ |\ i \in \omega\}$, binary relation symbols $=, \subseteq$, and unary relation symbols $'empty'$ and $'sing'$ (for `singleton').

\textbf{Atomic formulas} take the form $X_i=X_j$, $X_i \subseteq X_j$, $empty(X_i)$, $sing(X_i)$ or $R_i(X_{i_0},\dots,X_{i_{k-1}})$.
The \textbf{Monadic Second Order} $(=MSO)$ \textbf{formulas} are the closure of the atomic formulas by the usual logical connectives and second-order quantifiers $\exists X_i, \forall X_i$. An $\MSO$-formula is a \textbf{sentence} if it has no free variables.

We use upper case letters $X,Y,Z,\dots$ to denote variables; with an overline, $\bX, \bY$, etc. to denote finite tuples of variables (always assumed distinct). We denote by $|\bX|$ the length of a tuple $\bX$, that is, the number of variables appearing in it. A formula $\phi(X_0,\dots,X_{l-1})$ is a formula with free variables among $X_0,\dots,X_{l-1}$.

\item (Semantics) A \textbf{structure with} $l$ \textbf{distinguished elements} is a pair $(\mathcal{M},\bp)$ where $\mathcal{M}=(M,\langle R_i^\mathcal{M}\ |\ i \in I\rangle)$ is a structure and $\bp=(P_0,\dots,P_{l-1})$ is an $l$-tuple of subsets of $M$. 

Given a structure with $l$ distinguished elements $(\mathcal{M},\bp)$ and a formula $\phi(X_0,\dots,X_{l-1})$, we define the \textbf{satisfaction} relation $(\mathcal{M},\bp) \models \phi$ \big(also denoted by $\mathcal{M} \models \phi(\bp)$\big) as follows: $(\mathcal{M},\bp)\models X_i=X_j$ if $P_i,P_j$ are equal singletons, $(\mathcal{M},\bp) \models X_i \subseteq X_j$ if $P_i \subseteq P_j$, $(\mathcal{M},\bp) \models empty(X_i)$ if $P_i = \emptyset$, $(\mathcal{M},\bp) \models sing(X_i)$ if $P_i$ is a singleton, and $(\mathcal{M},\bp) \models R_i(X_{i_0},\dots,X_{i_{k-1}})$ if $P_{i_0},\dots,P_{i_{k-1}}$ are singletons $\{m_{i_0}\},\dots,\{m_{i_{k-1}}\}$, respectively, and $(m_{i_0},\dots,m_{i_{k-1}}) \in R_i^\mathcal{M}$. The boolean connectives are handled as usual and quantifiers range over subsets of $M$.

In weak-$\MSO$ semantics, the quantification is over finite sets.

\item Given a structure $\mathcal{M}$, the \textbf{monadic theory} of $\mathcal{M}$, $MTh(\mathcal{M})$, is the set of all sentences $\phi$ satisfied by $\mathcal{M}$.

\end{enumerate}
\end{defi}
Along the paper, when a calligraphic letter denotes a structure, its non-calligraphic uppercase form stands for the domain of that structure. All formulas in this paper are $\MSO$ formulas, unless otherwise noted.

\begin{defi} Let $\mathcal{M}=(M,\langle R_i^\mathcal{M}\ |\ i \in I\rangle)$ be a structure, and let $\bq$ be a tuple of subsets of $M$. A structure of the form $(M,\langle R_i^\mathcal{M}\ |\ i \in I\rangle,\bq)$ is called $\mathcal{M}$ \textbf{expanded by the parameters} $\bq$ (or simply $\mathcal{M}$ \textbf{with  parameters} $\bq$). It is denoted by $(\mathcal{M},\bq)$.
\end{defi}
\begin{defi}[first-order variables]
A free variable $x$ in a formula $\phi(x,\bY)$ is a \textbf{first-order free variable} if $\phi(x,\bY)$ is of the form $sing(x) \land \psi(x,\bY)$.
\end{defi}
We follow the convention that if a variable (or a tuple of variables) is denoted by a lowercase letter then it is a first order variable, unless otherwise noted. The structure $(\omega,<)$ is usually simply denoted by $\omega$.

Given a structure $\mathcal{M}$, one may ask whether $MTh(\mathcal{M})$ is a decidable set.

\begin{thm}[B\"uchi \cite{Buchi}] The monadic theory of $\omega$ is decidable.
\end{thm}
\begin{defi}
The $k$\textbf{-ary tree} is $T_k=k^{<\omega}$. The $k$\textbf{-ary tree with successors} $\mathcal{T}_k=(T_k,<,\bE)$ is $T_k$ with a binary relation $x<y$ if $x$ is a prefix of $y$, and unary predicates $Succ_i$ indicating sequences that end with $i$.
\end{defi}

\begin{thm}[Rabin \cite{Rabin69}] The monadic theory of $\mathcal{T}_k$ is decidable.
\end{thm}

\begin{defi}[Definability] Let $\mathcal{M}$ be a structure.

Given a formula $\psi(x_1,\dots,x_k)$, \[\psi^\mathcal{M}:=\{(x_1,\dots,x_k) \in M^k\ |\ \mathcal{M}\models \psi(\{x_1\},\dots,\{x_k\})\}.\]

A relation on $M$, $A \subseteq M^k$, is (weak-)
$\MSO$\textbf{-definable} in $\mathcal{M}$ if there is an (weak-)
$\MSO$ formula $\psi(x_1,\dots,x_k)$ such that $A=\psi^\mathcal{M}$.

We sometimes use `$\mathcal{M}$-definable' for `$\MSO$-definable in $\mathcal{M}$'.
\end{defi}
\begin{defi}[Mutual definability] Let $(M,\bar{R}), (M,\bar{S})$ be structures with the same domain $M$. We say that they are \textbf{mutually definable} if $\bar{R}$ is definable in $(M,\bar{S})$ and $\bar{S}$ is definable in $(M,\bar{R})$.
\end{defi}

\subsection{Elements of the composition method}\label{sect:comp}

This section presents the elements of the composition method. It is used only in the proofs of Theorems \ref{MainTheorem} and \ref{MainThmOrd} -- we advise to skip it now and read together with these proofs.
\subsubsection{Types} \label{TypesSub}
We review the basic definitions and facts about types, as may be found in \cite{ThomasComp}.

In what follows, $\mathcal{L}$ is a finite relational signature: $\mathcal{L}=\{R_1,\dots,R_n\}.$ The quantifier depth of an $\MSO$ formula is defined in a standard way. The set of formulas of quantifier depth $\leq d$ and free variables among $X_0,\dots,X_{l-1}$ is denoted by $\Formula{d}{l}$.
\begin{defi}
Let $\mathcal{M}_1,\mathcal{M}_2$ be $\mathcal{L}$-structures, and $\bar{A}_1 \subseteq M_1$, $\bar{A}_2 \subseteq M_2$ be $l$-tuples of subsets. We say that \[(\mathcal{M}_1,\bar{A}_1) \equiv_d (\mathcal{M}_2,\bar{A}_2)\] if they agree on $\Formula{d}{l}$.
\end{defi}

\begin{obs} For every $d$ and $l$, there are finitely many formulas in $\Formula{d}{l}$, up to equivalence.
\end{obs}

\begin{prop}[Hintikka formulas] For every $d,l \in \omega$, there is a finite set $H^d_l \subseteq \Formula{d}{l}$ such that:
\begin{enumerate}
\item For every $\mathcal{L}$-structure $\mathcal{M}$ and an $l$-tuple of subsets $\bar{A} \subseteq M$, there is a unique $\phi_d^{\mathcal{M},\bar{A}} \in H^d_l$ that is satisfied by $(\mathcal{M},\bar{A})$.
\item If $\phi \in H^d_l$ and $\psi \in \Formula{d}{l}$, then $\phi \models \psi$ or $\phi \models \neg \psi$.

Moreover, there is an algorithm that given such $\phi$ and $\psi$, decides between $\phi \models \psi$ and $\phi \models \neg \psi$ (it might be that both hold, and in this
case only one of the options is chosen).
\end{enumerate}
\end{prop}
The elements of $H^d_l$ are called $(d,l)$\textbf{-Hintikka formulas}.
\begin{cor}
Let $\mathcal{M}_1,\mathcal{M}_2$ be $\mathcal{L}$-structures, and $\bar{A}_1 \subseteq M_1$, $\bar{A}_2 \subseteq M_2$ be $l$-tuples of subsets. Then \[\phi_d^{\mathcal{M}_1,\bar{A_1}} \equiv \phi_d^{\mathcal{M}_2,\bar{A_2}}\] if and only if \[(\mathcal{M}_1,\bar{A}_1) \equiv_d (\mathcal{M}_2,\bar{A}_2).\]
\end{cor}

\begin{defi}[$d$-type of a structure]
Let $\mathcal{M}$ be an $\mathcal{L}$-structure, and $\bar{A} \subseteq M$ an $l$-tuple of subsets. The $d$\textbf{-type} of $(\mathcal{M},\bar{A})$ is $\phi_d^{\mathcal{M},\bar{A}}$. It is also denoted by $tp^d(\mathcal{M},\bar{A})$.
\end{defi}

Thus, the type $tp^d(\mathcal{M},\bar{A})$ effectively determines the formulas $\phi(\bX)$ of quantifier depth at most $d$ that hold in $(\mathcal{M},\bar{A})$.

\begin{defi} Let $\mathcal{L}=\{<\}$. An $\mathcal{L}$-structure with distinguished elements that is a linear order is called a \textbf{labelled linear order}. In other words, a labelled linear order is $(\mathcal{M},\bar{A})$ such that $\mathcal{M}=(M,<)$ is a linear order.
\end{defi}

\begin{defi}[Concatenation of labelled linear orders] Let $(\mathcal{M}_1,\bar{A}_1), (\mathcal{M}_2, \bar{A}_2)$ be labelled linear orders, where $\bar{A}_1$ and $\bar{A}_2$ are $l$-tuples of subsets. The labelled linear order $(\mathcal{M}_1,\bar{A}_1) + (\mathcal{M}_2, \bar{A}_2)$ is $(M,<,\bar{A})$ where:
\begin{enumerate}
\item The domain $M$ is the disjoint union of $M_1$ and $M_2$.
\item $< \ :=\ <^{\mathcal{M}_1} \cup <^{\mathcal{M}_2} \cup (M_1 \times M_2)$.
\item If $\bar{A}_1=(A_0^1,\dots,A_{l-1}^1)$ and $\bar{A}_2=(A_0^2,\dots,A_{l-1}^2)$, then $\bar{A}=(A_0^1 \cup A_0^2,\dots,A_{l-1}^1 \cup A_{l-1}^2)$.
\end{enumerate}
\end{defi}

\begin{prop} \label{EquivD} $\equiv_d$ is a congruence with respect to concatenation of labelled linear orders, which is, if $(\mathcal{M}_1,\bar{A}_1) \equiv_d (\mathcal{M}'_1,\bar{A}'_1)$ and $(\mathcal{M}_2,\bar{A}_2) \equiv_d (\mathcal{M}'_2,\bar{A}'_2)$, then \[\big((\mathcal{M}_1,\bar{A}_1) + (\mathcal{M}_2, \bar{A}_2)\big) \equiv_d \big((\mathcal{M}'_1,\bar{A}'_1) + (\mathcal{M}'_2, \bar{A}'_2)\big).\]
\end{prop}
That justifies the following definition:
\begin{defi}[Addition of types] \label{AdditionOfTypes}  Let $t,t_1,t_2$ be $d$-types of $l$ variables over $\mathcal{L}=\{<\}$.

We will say that $t=t_1+t_2$ if for any labelled linear orders $(\mathcal{M}_1,\bar{A}_1), (\mathcal{M}_2,\bar{A}_2)$, if \[tp^d(\mathcal{M}_1,\bar{A}_1)=t_1;\ tp^d(\mathcal{M}_2,\bar{A}_2)=t_2\] then \[tp^d((\mathcal{M}_1,\bar{A}_1)+(\mathcal{M}_2,\bar{A}_2))=t.\]
This is well defined due to Proposition \ref{EquivD}.
\end{defi}

The next  theorem is a standard consequence of composition methods for $\MSO$ over linear orders:

\begin{thm} \label{BinaryOrdinalParam}
Let $(\mathcal{M},\bar{A})$ be a labelled linear order, and let  $R(x,y)$ be a definable binary relation on $M$. Then there is $n \in \omega$ such that $R(x,y)$ is effectively determined by:
\begin{itemize}
\item $x=y$,
\item $x<_{\mathcal M}y$,
\item $tp^n(M\cap [-\infty,min\{x,y\}),<,\bar{A})$,
\item $tp^n(M\cap [min\{x,y\},max\{x,y\}),<,\bar{A})$,
\item $tp^n(M\cap [max\{x,y\},\infty),<,\bar{A})$.
\end{itemize}
\end{thm}

\subsubsection{Composition theorems for the $k$-ary tree} \label{CompSub} \hfill

In the following subsection, that is based on \cite{LifShelWO}, we consider the $k$-ary tree $T_k$ with its partial order $<$ and parameters $\bp \subseteq T_k$. The parameters \textbf{may or may not include} the successors $\bE$.

\begin{defi}
\begin{enumerate}
\item $P \subseteq T_k$ is a \textbf{path} if it is linearly ordered, and for every $x \in P$, if $x$ is not a maximum of $P$ then exactly one of its successors is in $P$.
\item $B\subseteq T_k$ is a \textbf{branch} if it is a maximal linearly ordered subset.
\item Let $z\in T_k$. The sub-tree $\{x\ |\ x \geq z\}$ is denoted by $T_{\geq z}$.
\end{enumerate}
\end{defi}

We provide a characterization of the definable binary relations of $(T_k,<,\bp)$ in terms of types of certain paths and sub-trees.

We use the notation $\mathcal{T}^d_l$ for the set of $d$-types with $l$ free variables over the signature $\{<\}$.

\begin{defi}[$d$-expansion of a path]\label{def:expansion-path}

Given $z < z' \in T_k$ and $x \in [z,z')$, the forest $T_{>x\perp [z,z')}$ is the union of $T_{\geq x \cdot i}$ for every $i$ such that $x \cdot i \not\in [z,z']$. Given $\tau \in \mathcal{T}^d_{l}$, let $Q^\tau_{z,z',\bp}$ be such that $x \in Q^\tau_{z,z',\bp}$ if \[tp^d(T_{>x\perp [z,z')},<,\bp)=\tau.\] By $\bq^{d}_{z,z',\bp}$ we refer to a vector of all predicates $Q^\tau_{z,z',\bp}$ for any $\tau \in \mathcal{T}^d_{l}$, where $l=|\bp|$.

The $d$\textbf{-expansion of the path} $[z,z')$ of the parameterized tree $(T_k,<,\bp)$ is \[([z,z'),<,\bp,\bq^{d}_{z,z',\bp}).\]
\end{defi}

\smallskip

A multi-set stands for a set with repetitions.

\begin{thm} \label{BinaryTkParam} (Definable binary relations on $(T_k,<,\bp)$) Let $R(x,y)$ be a definable binary relation on $(T_k,<,\bp)$. There are computable $n,d \in \omega$ such that for any $x,y \in T_k$ and $z(x,y)$ their largest common prefix, $R(x,y)$ is effectively determined by \[tp^n([\epsilon,z),<,\bp,\bq^{d}_{\epsilon,z,\bp}),\]\[ tp^n([z,x),<,\bp,\bq^{d}_{z,x,\bp}), tp^n(T_{\geq x},<,\bar{P}),\ \]\[ tp^n([z,y),<,\bp,\bq^{d}_{z,y,\bp}), tp^n(T_{\geq y},<,\bar{P}), \] and the multi-set of \[\ \{tp^n(T_{\geq z \cdot j},<,\bar{P})\ |\ (x \not \geq z \cdot j) \land (y \not \geq z \cdot j)\}.\]
\end{thm}

\subsection{Automata}

\begin{defi} Let $\Sigma$ be a finite alphabet.

A \textbf{deterministic parity automaton} on $\omega$-words over $\Sigma$ is $\mathcal{A}=(Q,q_0,\delta:Q\times\Sigma \to Q, c:Q \to \omega)$ where $Q$ is a finite set of states, $q_0 \in Q$ is an initial state, $\delta$ is a transition function, and $c$ is a coloring of the states.

Given an $\omega$-word $x$ over $\Sigma$, the \textbf{run} of $\mathcal{A}$ on $x$ is an $\omega$-word over $Q$, $(q_0,q_1,\dots)$, such that $q_0$ is the initial state, and for every $n$, $q_{n+1} = \delta(q_n,x(n))$.

An $\omega$-word $x$ over $\Sigma$ is \textbf{accepted} by $\mathcal{A}$ if in the run of $\mathcal{A}$ on $x$ the minimal state color that repeats infinitely many times is even. The \textbf{language} of $\mathcal{A}$, denoted by $\mathcal{L}(\mathcal{A})$, is the set of $\omega$-words that are accepted by $\mathcal{A}$.
\end{defi}

\begin{defi} \label{CharFun}
Let $M$ be a set, and $M_1,\dots,M_n$ be subsets of $M$. The \textbf{characteristic function} of $M_1,\dots,M_n$ is \[f_{M_1,\dots,M_n}: M \to 2^{\{1,\dots,n\}}\] defined by mapping each $m \in M$ to the indices of the subsets it belongs to.
\end{defi}

The classical B\"uchi-McNaughton theorem establishes the equivalence between Monadic Second-Order ($\MSO$) logic and deterministic parity automata over infinite words:

\begin{thm}[B\"uchi \cite{Buchi}, McNaughton \cite{McNaughton66}] \label{AutMSOEqOmega}
Let $\phi(X_1,\dots,X_n)$ be an $\MSO$ formula over $\omega$. There is a deterministic parity automaton $\mathcal{A}_\phi$ over $2^{\{1,\dots,n\}}$ such that for every $M_1,\dots,M_n \subseteq \omega$: \[\omega\models \phi(M_1,\dots,M_n)\] if and only if $\mathcal{A}_\phi$ accepts $f_{M_1,\dots,M_n}$.
Moreover, $\mathcal{A}_\phi$ is computable from $\phi$.
\end{thm}

\begin{defi} Let $\Sigma$ be a finite alphabet.

A \textbf{parity} $k$\textbf{-tree automaton} over $\Sigma$ is $\mathcal{A}=(Q,q_0,\Delta \subseteq Q\times\Sigma \times Q^k, c:Q \to \omega)$ where $Q$ is a finite set of states, $q_0 \in Q$ is an initial state, $\Delta$ is a transition relation, and $c$ is a coloring of the states.

A $\Sigma$\textbf{-labelled} $k$\textbf{-ary tree} is a function $\mu:T_k \to \Sigma$. Given a $\Sigma$-labelled $k$-ary tree $\mu:T_k \to \Sigma$, a \textbf{run} of $\mathcal{A}$ on $\mu$ is a $Q$-labelled $k$-ary tree $\nu$ in which the root is labelled by $q_0$, and for every $t \in T_k$, $(\nu(t),\mu(t),\nu(t\cdot 0),\dots,\nu(t\cdot (k-1))) \in \Delta$.

A $\Sigma$-labelled $k$-ary tree $\mu:T_k \to \Sigma$ is \textbf{accepted} by $\mathcal{A}$ if there is $\nu$ a run of $\mathcal{A}$ on $\mu$ such that on each branch of $\nu$, the minimal state color that repeats infinitely many times is even. The \textbf{language} of $\mathcal{A}$, denoted by $\mathcal{L}(\mathcal{A})$, is the set of $\Sigma$-labelled $k$-ary trees that are accepted by $\mathcal{A}$.
\end{defi}

\begin{thm}[Rabin \cite{Rabin69}] \label{AutLogEqTrees}
Let $\phi(X_1,\dots,X_n)$ be an $\MSO$ formula over $\mathcal{T}_k$. There is a parity $k$-tree automaton $\mathcal{A}_\phi$ over $2^{\{1,\dots,n\}}$ such that for every $M_1,\dots,M_n \subseteq T_k$: \[\mathcal{T}_k \models \phi(M_1,\dots,M_n)\] if and only if $\mathcal{A}_\phi$ accepts $f_{M_1,\dots,M_n}$. Moreover, $\mathcal{A}_\phi$ is computable from $\phi$.
\end{thm}

\subsection{Selection and uniformization}\label{sect:select}
\begin{defi}[Selection]
Let $\mathcal{M}$ be a structure. A formula $\psi(\bX)$ \textbf{selects} (or is a \textbf{selector} of) $\phi(\bX)$ over $\mathcal{M}$ if the following conditions hold:
\begin{itemize}
\item $\mathcal{M} \models \exists \bX \phi(\bX) \to \exists \bX \psi(\bX)$.
\item $\mathcal{M} \models \exists^{\leq 1} \bX \ \psi(\bX)$.
\item $\mathcal{M} \models \forall \bar{X} (\psi(\bar{X}) \to \phi(\bar{X})).$
\end{itemize}
(Here, $\exists^{\leq 1}$ stands for ``there exists at most one''.)

$\mathcal{M}$ has the \textbf{selection property} if every $\phi(\bX)$ has a selector over $\mathcal{M}$. Equivalently, $\mathcal{M}$ has the selection property if every satisfiable $\phi(\bX)$ has a definable model $\bX$.

$\mathcal{M}$ has the \textbf{solvable selection property} if there is an algorithm that produces selectors over $\mathcal{M}$.

\end{defi}

\begin{defi}[Selection over a class] Let $\mathcal{C}$ be a class of structures. A formula $\psi(\bX)$ \textbf{selects} $\phi(\bX)$ over $\mathcal{C}$ if for every $\mathcal{M} \in \mathcal{C}$, $\psi$ selects $\phi$ over $\mathcal{M}$.

$\mathcal{C}$ has the \textbf{selection property} if every $\phi(\bX)$ has a selector over $\mathcal{C}$.

$\mathcal{C}$ has the \textbf{solvable selection property} if there is an algorithm that produces selectors over $\mathcal{C}$.

\end{defi}
\begin{defi}[Uniformization] Let $\mathcal{M}$ be a structure. A formula $\psi(\bX,\bar{Y})$ \textbf{uniformizes} (or is a \textbf{uniformizer} of) $\phi(\bX,\bar{Y})$ with domain variables $\bX$ over $\mathcal{M}$ if the following conditions hold:
\begin{itemize}
\item $\mathcal{M} \models \forall \bX (\exists \bY \phi(\bX,\bY) \to \exists \bY \psi(\bX,\bY)).$
\item $\mathcal{M} \models \forall \bX \ \exists^{\leq 1} \bY \ \psi(\bX,\bY)$.
\item $\mathcal{M} \models \forall \bX \ \forall \bar{Y} (\psi(\bX,\bar{Y}) \to \phi(\bX,\bar{Y})).$
\end{itemize}
$\mathcal{M}$ has the \textbf{uniformization property} if every $\phi(\bX,\bar{Y})$ has a formula that uniformizes it with domain variables $\bX$ over $\mathcal{M}$.

$\mathcal{M}$ has the \textbf{solvable uniformization property} if there is an algorithm that produces uniformizers over $\mathcal{M}$.

\end{defi}
\begin{defi}[Uniformization over a class] Let $\mathcal{C}$ be a class of structures. A formula $\psi(\bX,\bar{Y})$ \textbf{uniformizes} $\phi(\bX,\bar{Y})$ with domain $\bX$ over $\mathcal{C}$ if for every $\mathcal{M} \in \mathcal{C}$, $\psi$ uniformizes $\phi$ with domain $\bX$ over $\mathcal{M}$.

$\mathcal{C}$ has the \textbf{uniformization property} if every $\phi(\bX,\bY)$ has a formula that uniformizes it with domain $\bX$ over $\mathcal{C}$.

$\mathcal{C}$ has the \textbf{solvable uniformization property} if there is an algorithm that produces uniformizers over $\mathcal{C}$.
\end{defi}
An immediate observation is that a structure $\mathcal{M}$ has the uniformization property if and only if for every $l$, the class of $\mathcal{M}$ with $l$ parameters has the selection  property

\begin{thm}[\cite {LifschesS98, AlexShomratUnif}]\label{unifoAlpha}
Let $\alpha$ be an ordinal. $(\alpha,<)$ has the solvable uniformization property if and only if $\alpha < \omega^\omega$. In particular, for $\alpha < \omega^\omega$ and $\bp \subseteq \alpha$, $(\alpha,<,\bp)$ has solvable selection and solvable uniformization.
\end{thm}

As to the $k$-ary tree $T_k$, when it is only equipped with its partial order, selection fails for the formula ``$x$ is a successor of the root''. However, by Rabin's basis theorem:
\begin{thm}[Rabin \cite{Rabin69}]
$\mathcal{T}_k=(T_k,<,\bE)$  has the solvable selection property.
\end{thm}

As to uniformization:
\begin{thm}[Gurevich-Shelah \cite{ShelahGurevich}] $\mathcal{T}_k$ doesn't have the uniformization property.
\end{thm}

When $\mathcal{T}_k=(T_k,<,\bE)$ is expanded by parameters, selection may or may not hold:

\begin{exa}[\cite{Fratani} Level predicates] \label{YesSelExample} Let $P \subseteq \omega$, and let $P^*\subseteq T_k$ be defined by \[x \in P^* \iff |x| \in P.\] $P^*$ is called a \textbf{level predicate}. Let $(\mathcal{T}_k,\bar{P^*})$ be an expansion of $\mathcal{T}_k$ by some level predicates $\bar{P^*}$. $(\mathcal{T}_k,\bar{P^*})$ has the selection property. Moreover, the class of subtrees $(T_{\geq z},<,\bE,\bar{P^*})$ has the selection property.
\end{exa}
\begin{restatable}[\cite{carayol2010choice}]{prop}{NoSelExampleStatement}\label{NoSelExample}  There is $P \subseteq T_2$ such that $(\mathcal{T}_2,P)$ doesn't have the selection property.
\end{restatable}

\subsection{Copying}
\begin{defi}
Let $\mathcal{M}$ be a structure for the signature $\mathcal{L}=\langle R_i\ |\ i \in I \rangle$, and $d \in \omega$. The $d$\textbf{-copying} of $\mathcal{M}$, $\mathcal{M} \times d$, is a structure for the signature $\mathcal{L} \times d=\mathcal{L} \cup \{\sim,P_0,\dots,P_{d-1}\}$ defined as follows:
\begin{itemize}
\item The domain is the cartesian product $M \times \{0,\dots,d-1\}$.
\item A tuple $(m_0,j_0),\dots,(m_l,j_l)$ is in $R_i$ if $(m_0,\dots,m_l) \in R_i$ and $j_0=\cdots=j_l$.
\item $(m,j)\in P_k$ if $j=k$.
\item $(m_0,j_0) \sim (m_1,j_1)$ if $m_0=m_1$.
\end{itemize}
\end{defi}

\begin{lem} \label{copyingLem} Let $\mathcal{L}$ be a signature and $d \in \omega$.

For every $\phi(\bX) \in \MSO(\mathcal{L} \times d)$ there is a computable $\phi^*(\bX_0,\dots,\bX_{d-1}) \in \MSO(\mathcal{L})$ \large($|\bX_i|=|\bX|$\large), such that for any $\bA \subseteq M \times \{0,\dots,d-1\}$ \[\mathcal{M} \times d \models \phi(\bA) \iff \mathcal{M} \models \phi^*(\bA_0,\dots,\bA_{d-1})\] where $\bar{A_i}$ is the intersection $\bA \cap (M \times \{i\})$.

And vice versa: for every $\phi^*(\bX_0,\dots,\bX_{d-1}) \in \MSO(\mathcal{L})$ there is a computable $\phi^{**}(\bX) \in \MSO(\mathcal{L} \times d)$ such that for any $\bA \subseteq M \times \{0,\dots,d-1\}$ \[\mathcal{M} \times d \models \phi^{**}(\bA) \iff \mathcal{M} \models \phi^{*}(\bA_0,\dots,\bA_{d-1})\] where $\bar{A_i}$ is the intersection $\bA \cap (M \times \{i\})$.
\end{lem}
\begin{cor} \label{copyingRecursive}
For every $d\in\omega$, $MTh(\mathcal{M})$ and $MTh(\mathcal{M} \times d)$ are mutually recursive.
\end{cor}

\begin{prop} \label{copyingSelection}
$\mathcal{M}$ has the (solvable) selection property if and only if $\mathcal{M} \times d$ has the (solvable) selection property.
\end{prop}

\subsection{Games} \label{preliminary_games}
\begin{defi} \begin{enumerate} \item A \textbf{graph} is $G=(V,E)$ where $V$ is a set of vertices, and $E$ is a binary relation on $V$ representing directed edges, which is, $(v_1,v_2)\in E$ represents an edge from $v_1$ to $v_2$.
 \item Let $v_1,v_2\in V$. If $(v_1,v_2)\in E$ we will say that $v_2$ is an \textbf{out-neighbour} of $v_1$ and $v_1$ is an \textbf{in-neighbour} of $v_2$. A vertex $v_1$ is a \textbf{neighbour} of $v_2$ if either $(v_1,v_2) \in E$ or $(v_2,v_1)\in E$.
 \item A \textbf{path} in a graph is a finite or infinite sequence of vertices $\langle v_n\ |\ n \in N\rangle$ for $N\in\omega\cup\{\omega\}$, such that $(v_i,v_{i+1}) \in E$ whenever $i$ and $i+1$ are both in $N$.
\end{enumerate}
\end{defi}
\begin{defi}
\begin{enumerate}
\item A \textbf{parity game} $G$ is $(V_I,V_{II},E,\Pi_0,\dots,\Pi_{d-1})$ where $V_I,V_{II}$ are disjoint sets of vertices, $E$ is an edge relation on $V=V_I \cup V_{II}$ such that the out degree of each vertex is at least $1$, and $\Pi_0,\dots,\Pi_{d-1}$ is a partition of $V$ into $d$ different colors.
\item Let $v_0 \in V_I \cup V_{II}$. In a play of the game from $v_0$ the two players begin their play in the initial position $v_0$. On each round, if the position is $v \in V_I$ then player $I$ chooses the next position $v'$ such that there is an edge between $v$ and $v'$. Otherwise, if $v \in V_{II}$, player $II$ is the one choosing the next position $v'$ such that $E(v,v')$.

The resulting play of the game is then an $\omega$-path in $G$. The \textbf{play} is \textbf{winning} for player $I$ if the minimal $i$ such that infinitely many vertices of the play are in $\Pi_i$ is even. Otherwise, the play is winning for player $II$.

\item A \textbf{strategy} for player $I$ is $s:(V^*\cdot V_I) \to V$. A play of the game $\langle v_n\ |\ n \in \omega \rangle$ is played by the strategy $s$ if for every $n$, if $v_n \in V_I$ then $v_{n+1} = s(v_0,\dots,v_n)$.

\item A \textbf{memoryless (=positional) strategy} for player $I$ is $s:V_I \to V$. A play of the game $\langle v_n\ |\ n \in \omega \rangle$ is played by the strategy $s$ if for every $n$, if $v_n \in V_I$ then $v_{n+1} = s(v_n)$.

\item A \textbf{strategy} $s$ for player $I$ is \textbf{winning} from a vertex $v$ if all game plays from $v$ that are played by the strategy $s$ are winning for player $I$. The set of vertices from which a strategy $s$ is winning is the \textbf{winning region of} $s$.
\item The \textbf{winning region of player} $I$ is the union of the winning regions of all strategies for player $I$.
\item A winning strategy for player $I$ is \textbf{uniform} if its winning region is the winning region of player $I$.
\item All the above are defined for player $II$ in a similar fashion.
\end{enumerate}
\end{defi}
The following summarizes some basic facts about parity games (for further reading, see \cite{Gradel,Pin}):
\begin{thm}
Let $G$ be a parity game.
\begin{enumerate}
\item For each $v \in V$, if a player has a winning strategy from $v$ then it has a memoryless winning strategy from $v$.
\item Each $v \in V$ is either in the winning region of player $I$, or in the winning region of player $II$.
\item Both players have uniform memoryless winning strategies.
\item The winning regions of both players are $\MSO$-definable in $(V_I,V_{II},E,\Pi_0,\dots,\Pi_{d-1})$.
\end{enumerate}
\end{thm}

\begin{defi}
Let $\mathcal{M}$ be a structure. A parity game $G$ is $\MSO$\textbf{-definable}
 in $\mathcal{M}$ if
$V_I,V_{II},\Pi_0,\dots,\Pi_{d-1}$ are $\MSO$-definable subsets of $\mathcal{M}$, and $E$ is an $\MSO$-definable binary relation on $\mathcal{M}$. Similarly, a memoryless strategy is $\MSO$-definable in $\mathcal{M}$ if it is an $\MSO$-definable binary relation on $\mathcal{M}$.
\end{defi}

What we call "$\MSO$-definable" is often called "one-dimensional $\MSO$ interpreted".

In Sections \ref{SubsectionApplications} and \ref{SubsectionAbstractStore} we provide examples of such games.

If a parity game is $\MSO$-definable in $\mathcal{M}$, the winning regions are $\MSO$-definable. However, as Proposition \ref{FailureT2P} will show, the existence of $\MSO$-definable uniform memoryless winning strategies is not guaranteed.

\section{bounded degree games} \label{SectionBoundedOutDegree}

\subsection{Definable strategies in bounded degree games} \label{SubsectionBoundedOutDegree}

In the following subsection, we find a sufficient condition on $\mathcal{M}$ such that parity games with bounded degree that are definable in $\mathcal{M}$ will have definable uniform memoryless winning strategies. We will see why bounded degree, and even bounded out degree, simplifies the description of a memoryless strategy, up to the point that the selection property finds a definition of a uniform memoryless winning strategy.
\begin{defi}
\begin{enumerate}
\item We will say a game has \textbf{bounded degree} if its arena has bounded degree, which is, there is $d$ such that all vertices have at most $d$ neighbours.

A game has \textbf{bounded out degree} if there is $d$ such that all vertices of its arena have at most $d$ out-neighbours.
\item A ternary relation $R(u,a,b)$ is a \textbf{local linear order} of a graph $G=(V,E)$ if for any vertex $u\in V$, $R_u(a,b)=\{(a,b)\ |\ R(u,a,b)\}$ linearly orders all out-neighbours of $u$.
\item A binary relation $R(a,b)$ is a \textbf{uniform local linear order} of a graph $G=(V,E)$ if for any vertex $u$ it linearly orders all out-neighbours of $u$.
\item Let $G$ be an $\mathcal{M}$-definable graph. We will say that $G$ has a \textbf{definable} (uniform) local linear order if there is an $\mathcal{M}$-definable relation that is a (uniform) local linear order of $G$.
\end{enumerate}
\end{defi}
When a graph has out degree bounded by $d$ and a definable local linear order $R(u,x,y)$, a $d$-tuple of monadic predicates $(P_0,\dots,P_{d-1})$ that forms a partition of the vertices of one of the players defines a positional strategy: given a vertex $u$, the player orders the out-neighbours of $u$ by $R_u$, and if $P_i(u)$ then the next move is the $i$'th vertex.

\smallskip
The following provides sufficient conditions for the existence of a definable local linear order:
\begin{restatable}{lem}{DefLocLinearOrderStatement} \label{DefLocLinearOrder}
If $G$ is a bounded degree graph definable in $\mathcal{M}$, and $\mathcal{M}$ has the selection property, then $G$ has an $\mathcal{M}$-definable uniform local linear order.

Moreover, if selection for $\mathcal{M}$ is solvable then the local linear order is computable.
\end{restatable}
\begin{proof}
This is based on the following well-known coloring theorem:
\begin{prop}[De Bruijn-Erdos \cite{DeErd}]
If $G$ is a graph whose vertices have degree at most $d$, then there is a $d+1$ coloring of the graph, which is, a coloring with $d+1$ colors such that an edge connects vertices of different colors.
\end{prop}
\begin{proof}
For finite graphs, this is proven by induction on the number of vertices - if $v \in G$ and $G \setminus v$ is $d+1$ colored, then since $v$ has at most $d$ neighbours, it may be assigned a color different from all of its neighbours.

For infinite graphs, it follows from compactness of propositional logic.
\end{proof}
Given $G=(V,E)$ a graph with bounded degree $d$, we define a new edge relation on the same vertices: \[(v,u) \in E' \iff (v \neq u) \land \exists w\ \big((w,v) \in E \land (w,u)\in E\big).\] The resulting graph has degree at most $d^2$, so there exists a $d^2+1$ coloring. This coloring induces a partition of $V$ into sets $C_1,\dots,C_{d^2+1}$. On the original graph $G$, this partition satisfies:\[\phi(C_1,\dots,C_{d^2+1}) = \forall v,u_1,u_2\] \[ (v,u_1)\in E \land (v,u_2) \in E \to c(u_1) \neq c(u_2)\]
(where $c(u_1)\neq c(u_2)$ is formalizable in $\MSO(C_1,\dots,C_{d^2+1})$).

We have so far shown that the formula $\phi(\bar{C})$ is satisfiable. Since $\mathcal{M}$ has the selection property, there is an $\mathcal{M}$-definable $\bar{C}$ that satisfies the formula. Using $\bar{C}$ a uniform local linear order may be defined: \[a<b \iff \bigvee _{1\leq i < j \leq (d^2+1)}\ C_i(a) \land C_j(b).\]
\end{proof}

\smallskip
\begin{thm} \label{boundedOutDegree}
If $\mathcal{M}$ has the \textbf{selection property}, and $G$ is a parity game definable in $\mathcal{M}$ such that $G$ has \textbf{bounded out degree} and a \textbf{definable local linear order}, then there are $\mathcal{M}$\textbf{-definable uniform memoryless winning strategies} for $G$.

Moreover, if selection for $\mathcal{M}$ is solvable then the winning strategies are computable.
\end{thm}

\begin{proof}

The proof has $3$ steps:
\begin{enumerate}
\item Code a uniform memoryless strategy as a tuple $\bar{P}$ of unary predicates.
\item Write a formula $uWinSt(\bar{P})$ saying that $\bar{P}$ codes a uniform memoryless winning strategy.
\item Apply selection on the formula $uWinSt(\bar{P})$ to get a definable uniform memoryless winning strategy.
\end{enumerate}

Let $k$ be the number of colors of the vertices, and $d$ a bound of the out degree.
Let \[V_I(x),V_{II}(x),E(x,y),\Pi_0(x),\dots,\Pi_{k-1}(x)\] be formulas that define $G$ in $\mathcal{M}$. Let $\bp = (P_0,\dots,P_{d-1}), \bq = (Q_0,\dots,Q_{d-1})$ be unary predicates. Let $R(u,a,b)$ be an $\mathcal{M}$-definable local linear order of $G$.

\begin{itemize}
\item Let $p_I(\bp),p_{II}(\bq)$ be formulas saying that $\bp$ and $\bq$ are partitions of $V_I$ and $V_{II}$, respectively.

\item Let $st_I(x,y,\bp)$ be a formula saying that $x\in V_I$, and for any $0 \leq i \leq d-1$, if $P_i(x)$ and the out-neighbours of $x$ ordered by $R_x$ are $y_0 < \dots < y_k$ (for $k \leq d-1$), then $y=y_i$. The formula $st_{II}(x,y,\bq)$ is defined similarly.

\item The following formula $reach(x,\bp,\bq,Y)$ says that playing from $x$ by the strategies $\bp,\bq$ leads to a vertex of the set $Y$:
\[\forall A\ \Big((x \in A) \land \] \[\big(\forall z \in (A \cap V_I)\ \forall y \ st_I(z,y,\bp) \to y \in A\big) \land \]
\[\big(\forall z \in (A \cap V_{II})  \ \forall y \ st_{II}(z,y,\bq) \to y \in A\big) \Big) \]\[\to (A \cap Y \neq \emptyset). \]

\item The formula $win_I(x,\bp,\bq)$ says that playing from $x$ by the strategies $\bp,\bq$ results in a win of player I:

There is an even color $i$ such that:
\begin{itemize}
\item For every $j<i$, there is $y$ satisfying $reach(x,\bp,\bq,\{y\}) \land \neg reach(y,\bp,\bq,\Pi_j)$.
\item For every $y$, if $reach(x,\bp,\bq,\{y\})$ then $reach(y,\bp,\bq,\Pi_i)$.
\end{itemize}
\item The formula $WinSt_I(x,\bp)$ is \[p_I(\bp) \land \big(\forall \bq\ p_{II}(\bq) \to win_I(x,\bp,\bq)\big),\] which is, $\bp$ codes a winning strategy of player I from $x$.

\item This leads to a definition $winReg_I(x)$ of the winning region of player I: $x$ is in the winning region of player I if there is $\bp$ such that $WinSt_I(x,\bp)$.

\item Finally, $uWinSt_I(\bp)$, the set of codes of uniform winning strategies for $I$, is \[\forall x\ winReg_I(x) \to WinSt_I(x,\bp).\]
\end{itemize}
Since $\mathcal{M} \models \exists \bp \ uWinSt_I(\bp)$, and $\mathcal{M}$ has the selection property, there is a formula $uWinSt^*_I(\bp)$ such that there is a unique $\bp$ satisfying $uWinSt_I^*(\bp)$, and this $\bp$ also satisfies $uWinSt_I(\bp)$. The formula \[\exists \bp \ uWinSt_I^*(\bp) \land st_I(x,y,\bp)\] defines a uniform memoryless winning strategy for player I.

The existence of a definable uniform memoryless winning strategy for player $II$ is proven similarly.
\end{proof}
Together with Lemma \ref{DefLocLinearOrder}:
\begin{cor} \label{boundedDegree}
If $\mathcal{M}$ has the \textbf{selection property}, and $G$ is a parity game definable in $\mathcal{M}$ such that $G$ has \textbf{bounded degree}, then there are $\mathcal{M}$\textbf{-definable uniform memoryless winning strategies} for $G$.

Moreover, if selection for $\mathcal{M}$ is solvable then the winning strategies are computable.
\end{cor}

\subsection{Applications} \label{SubsectionApplications}

\begin{exa}
Let $f: \mathbb{N} \to \mathbb{N}$ be the iterated factorial function: $f(n)=n!\dots !(n\ times)$. Let $P\subseteq \mathbb{N}$ be the range of $f$, and consider the structure $\mathcal{M}=(\mathbb{N},<,P)$.
By Elgot-Rabin \cite{ElgRabin}, MSO($\mathcal{M}$) is decidable and $\mathcal{M}$ has the selection property.

Our results show that a bounded out degree game $MSO$-definable in $\mathcal{M}$ has $\mathcal{M}$-definable uniform memoryless winning strategies. This does not follow from any previously known synthesis result.

Later we will see that the same applies for games of unbounded out degree.
\end{exa}

Corollary \ref{boundedDegree} may be applied on bounded out degree games definable in $\mathcal{T}_k$. An important example of games definable in $\mathcal{T}_k$ are pushdown parity games. These games have bounded degree. In fact, in these games the $\mathcal{T}_k$-distance between neighbours is bounded:
\begin{defi}
Let $d \in \omega$. A graph $G=(V,E)$, $V \subseteq T_k$, is $d$\textbf{-local} if for every $x,y \in T_k$, if $E(x,y)$ then $y$ is the result of at most $d$ operations of successor and predecessor on $x$.

$G=(V,E)$ is d-\textbf{locally definable} in $\mathcal{T}_k$ if it is $d$-local and definable in $\mathcal{T}_k$. $G$ is \textbf{locally definable} if it is d-\textbf{locally definable} for some $d\in \omega$.
\end{defi}

By Theorem \ref{boundedOutDegree}:
\begin{cor} \label{BddLocalInterp}
Let $G$ be a parity game locally definable in $\mathcal{T}_k$. Then $G$ has computable uniform memoryless winning strategies locally definable in $\mathcal{T}_k$.
\end{cor}

That may be applied to solve pushdown parity games:

\begin{restatable}[\cite{Wal01},\cite{CachatWaluk} Theorem 4.2]{thm}{WalukStatement} \label{WalukThm} A pushdown parity game has computable uniform winning pushdown strategies.
\end{restatable}

\begin{proof}
A pushdown parity game with stack alphabet $\Gamma$ and states $P$ can be defined in the copying $\mathcal{T}_{|\Gamma|} \times P$. By Proposition \ref{copyingSelection} and Corollary \ref{boundedDegree},  $\mathcal{T}_{|\Gamma|} \times P$-definable uniform memoryless winning strategies can be computed. Moreover, after the addition of some definable predicates, these strategies are given by simple formulas. The strategies may be translated into uniform pushdown strategies.
\end{proof}

Note that this proof doesn't provide the complexity analysis provided by Walukiewicz.
 
\smallskip

Other applications of Corollary \ref{boundedDegree} are when $\mathcal{T}_k$ is expanded by parameters for which the selection property holds:

\begin{exa}[Level predicates]
Let $\mathcal{T}_k$ be expanded by $x\in \mathsf{Factorial}$ if $|x|$ is in $\{n!\ |\ n \in \omega\}$. Then bounded out degree games definable in the resulting structure have definable strategies.

In fact, whenever $\mathcal{T}_k$ is expanded by level predicates (see Example \ref{YesSelExample}), its definable bounded out degree games also have definable strategies.
\end{exa}

Higher-order pushdown games \cite{CarayolHMOS08} and iterated store pushdown games \cite{Fratani} are definable in expansions of $\mathcal{T}_k$ by parameters for which the selection property holds \cite{Fratani}. It follows that in these cases as well, definable strategies exist. (See also \cite{CarayolS08} for effective constructions of strategies.)

 \subsection{Games over an  Abstract Store} \label{SubsectionAbstractStore}
 In this subsection we consider games over an abstract store. Abstract store generalizes pushdown store and higher-order pushdown store.
 We obtain a generalization of Walukiewicz's theorem as a consequence of the results of subsection \ref{SubsectionBoundedOutDegree}.
\begin{defi}[Abstract store]
An abstract store is a tuple ${\As}:=({\AStore},\bar{P}, OP
)$, where $\AStore$ is a set, $ \bar{P}$ is a tuple of unary predicates on $\AStore$, and $OP $ is a finite set of abstract store operations which are partial functions
from $ \AStore$ to $\AStore$.
 \end{defi}

For example, for every finite  alphabet $\Gamma$,  the pushdown store over $\Gamma$ can be considered an instance of an abstract store,
 where ${\AStore}:=\Gamma^*$ is the set of finite strings over $\Gamma$, for every $\gamma\in \Gamma$ we have a predicate $P_\gamma$ which holds on the stores with top letter $\gamma$,
 for every $\gamma\in \Gamma$ the operation $push_\gamma$  pushes $\gamma$ into the stack, and if the stack $s$ is not empty the operation $pop$ removes its top letter.

Higher-order pushdown automata and register automata can also be viewed as instances of this definition.

 An abstract store ${\As}=({\AStore},\bar{P}, OP
 )$ can be seen as a labelled graph.
 The set of nodes is $\AStore$. The node labelling is $\bar{P}$.
 There is an edge from $u$ to $v$ labelled by $f\in OP$ if $f(u)=v$.

 \begin{defi}[An automaton over an abstract store]  Let ${\As}=({\AStore},\bar{P}, OP
 )$ be an abstract store,  such that $|\bar{P}|=l$. An automaton $\cA$ over $\As$ is a tuple
 ${\cA}:=(Q, \Delta
 )$  where $Q$ ($=Q_{\cA}$) is a finite set of states and $ \Delta \subseteq Q \times \{0,1\}^l \times OP \times Q$ is a transition relation.
 \end{defi}
 A \emph{configuration }is a pair $(q,s)\in Q\times \AStore$. We define a labelled transition relation  $\cA\times \As$ over the configurations:
 $(q,s)\rightarrow_f (q',s')$ if $f\in OP$, $s'=f(s)$, $\bar{P}(s)=\bar{b}$ and $ (q, \bar{b},f,q')\in \Delta$.

 \begin{defi}[An abstract store game]  Let $\As$ be an abstract store, let $\cA$ be an automaton over $\As$, and let the set $Q$ of states of $\cA$ be partitioned into $Q_I, Q_{II}$ - the sets of the first and the second player states, and into $Q_1, \dots, Q_d$ - the set of states with color $i:=1, \dots , d$. The abstract store parity  game over $\cA$ and $\As$ induced by these partitions
 is the game over the transition system $\cA\times \As$  where a configuration $(q,s)$ is a node of the first (respectively, the second) player if $q\in Q_I$ (respectively, $q\in Q_{II}$), and the color of $(q,s)$ is the same as the color of $q$.
  \end{defi}
 Note that the above game is definable by very simple formulas in $|Q_{\cA}|$-copying of $\As$.

A memoryless winning strategy for Player II maps every configuration of player II to a configuration of Player I.

By Proposition \ref{copyingSelection} and Corollary \ref{boundedDegree}:

  \begin{cor} \label{Abstr-game}
If an abstract store $\As$ has the selection property, and $G$ is a parity game  over $\As$,  then there are   definable (in a copying of $\As$) uniform memoryless winning strategies for $G$.
Moreover, if selection for $\As$ is solvable then the winning strategies are computable.
\end{cor}
This may be rephrased in terms of abstract store strategies, which is a special type of automata over abstract stores that compute a strategy:

\begin{defi}[Abstract store strategies] Let ${\As}:=({\AStore},\bar{P}, OP )$ be an abstract store, and $\mathcal{A}:=(Q,\Delta)$ an automaton over $\As$. Let $G$ be the abstract store game over $\mathcal{A}\times \As$.
\begin{enumerate}

\item An \textbf{abstract store strategy} of player $I$ for the game $G$ is $(Q',{\As}',\Pi)$, where:
\begin{itemize}
\item $Q'$ is a finite set of states.
\item ${\As}'=({\AStore}',\bar{P'}, OP')$ is an abstract store such that $|\bar{P'}|=l$.
\item $\Pi$ is a union of transition functions:  \[\Pi_{I}: Q' \times \{0,1\}^l \to OP' \times Q' \times \Delta_I \]\[\mbox{\Big(update state and store and choose a transition\Big)}\] \[\Pi_{II}: Q' \times \{0,1\}^l \times \Delta_{II} \to OP' \times Q' \]\[\mbox{\Big(read Player II transition and update state and store\Big)}\] where $\Delta_I$ are transition rules in $\Delta$ whose first coordinate is in $Q_I$, and $\Delta_{II}$ is defined similarly.
\end{itemize}
An \textbf{abstract store strategy configuration} is $(q',s') \in Q' \times \mathcal{S}'$.
\item Given a player $I$ node $(q,s)$ and a transition rule $\delta=(q,\bar{P}(s),f,q')\in \Delta_I$, we say that \textbf{player} $I$ \textbf{plays the transition rule} $\delta$ if the next game position is $(q',f(s))$. Similarly for player $II$.
\item A play of $\mathcal{G}$ from a configuration $(q,s) \in Q \times \mathcal{S}$ is \textbf{played by the abstract store strategy} of player $I$ $(Q',{\As}',\Pi)$ from configuration $(q'_0,s'_0)\in Q' \times \mathcal{S}'$ if:
\begin{itemize}
\item The initial abstract store strategy configuration is $(q'_0,s'_0)$.
\item Whenever the vertex is of player $I$, the abstract store strategy configuration is $(q',s')$, and $\Pi_{I}(q',\bar{P'}(s'))=(f',q'',\delta)$, player $I$ plays the transition rule $\delta$ and updates the abstract store strategy configuration to $(q'',f'(s'))$.
\item Whenever the vertex is of player $II$, the abstract store strategy configuration is $(q',s')$, player $II$ plays $\delta \in \Delta_{II}$, and $\Pi_{II}(q',\bar{P'}(s'),\delta)=(f',q'')$, the abstract store strategy configuration is updated to $(q'',f'(s'))$.
\end{itemize}
\item An abstract store strategy of player $I$ is a \textbf{uniform winning abstract store strategy} if for every vertex $(q,s)$ in the winning region of player $I$, there is an abstract store strategy configuration $(q',s')$ such that any play of $\mathcal{G}$ from $(q,s)$ played by the abstract store strategy from $(q',s')$ is winning for player $I$.
\item Strategies of player $II$ are defined in a similar fashion.
\end{enumerate}
\end{defi}

 \begin{cor}[Synthesis of Abstract Store Strategies]
If ${\As}=({\AStore},\bar{P}, OP)$ is an abstract store with the selection property, $\mathcal{A}=(Q,\Delta)$ is an automaton over $\As$, and $G$ is the parity game over $\mathcal{A}$ and $\As$, then both players have uniform winning abstract store strategies $(Q,{\As}_1,\Pi_1), (Q,{\As}_2,\Pi_2)$ (with the same states as $\mathcal{A}$) such that ${\As}_1,{\As}_2$ are definable in $\As$.
Moreover:
\begin{enumerate}
\item The abstract stores ${\As}_1,{\As}_2$ are expansions of $\As$ by definable unary predicates.
\item If selection for $\As$ is solvable then the uniform winning abstract store strategies are computable.
\end{enumerate}
\end{cor}

\section{Unbounded degree games}\label{SectionUnbounded}

The following section continues the investigation of the following problem:
\begin{prob}
Let $\mathcal{M}$ be a structure, and $G$ a parity game definable in $\mathcal{M}$. Are there $\mathcal{M}$-definable uniform memoryless winning strategies?
\end{prob}
In section \ref{SectionBoundedOutDegree}, we've seen one case in which the answer is positive - when $\mathcal{M}$ has the selection property and $G$ has bounded degree (Corollary \ref{boundedDegree}).

In the following section we consider unbounded degree games, and present a general scheme for finding definable strategies in $\mathcal{M}$-definable unbounded degree games:
\begin{enumerate}
\item Represent the game by a bounded degree game defined in a copying of $\mathcal{M}$.
\item Find definable uniform memoryless winning strategies for the representing game.
\item Transfer these strategies into winning strategies for the original game.
\end{enumerate}
As long as $\mathcal{M}$ has the selection property, section \ref{SectionBoundedOutDegree} guarantees that step (2) may be performed, whereas step (3) does not require any new techniques. In the following section we will explain how to perform step (1).

In order to construct the required representation, we introduce the notion of definability by regular expressions. We use this notion to represent by bounded out degree games all unbounded degree games defined in $(\alpha,<,\bp)$ for $\alpha$ an ordinal below $\omega^\omega$, and also to represent games defined in $\mathcal{T}_k$ expanded by parameters. Our main technical result follows (Theorem \ref{MainDefStrategy}).

The structure of the section is as follows:

In subsection \ref{representationSub} the notion of representation is defined, and in particular the notion of copy representation, in which an $\mathcal{M}$-definable game is represented in a copying of $\mathcal{M}$. Variants of this  notion appear implicitly in the literature (see e.g. \cite{BrutschT22,CachatWaluk}). Although we're only interested in representing games, the challenging part is representing their graphs.

We develop some simple technical lemmas that construct a game representation out of a graph representation.
This will allow us later to concentrate on the graph
representation only.

In subsection \ref{MainSchemeSub}, the correctness of the general scheme is proven.

In subsection \ref{RegExpSub}, the notion of definability by regular expressions is introduced. It is shown that definability by regular expressions is an instance of $\MSO$ definability, and that if a graph is definable by regular expressions in a bounded degree graph then it is copy represented by a bounded degree graph. For the expansions of $\mathcal{T}_k$ and $\alpha<\omega^\omega$ by parameters, we show in subsection \ref{subsectMsovs-reg} that $\MSO$ definability is the same as definability by regular expressions in a bounded degree graph. That is used to represent graphs defined in these structures, and the main technical result -- Theorem \ref{MainDefStrategy} -- follows.

\subsection{Representation of graphs} \label{representationSub}
We begin with a motivating example that demonstrates in what sense a game may represent another game:
\begin{exa}
Let $G=(Even, Odd, E, \Pi_0, \Pi_1)$ be a parity game on $\omega$, where $\Pi_0,\Pi_1$ is some coloring, and $E(x,y)$ if $y>x$ and $y-x$ is either $1$ or it is a multiple of $3$.
\begin{itemize}
\item $G$ has unbounded out degree.
\item Let $\tilde{G}$ be a game on $\omega \cup \big(\omega \times \{I,II\}\big)$ defined as follows:
\begin{itemize}
\item The vertices of player $I$ are $Even \cup \big(\omega \times \{I\}\big)$; all the rest are of player $II$.
\item Colors: $\omega \times \{I\}$ is colored by $1$; $\omega \times \{II\}$ is colored by $0$; the rest of vertices get their previous color + 2.
\item The edges $\tilde{E}$ are divided into two types:
\begin{enumerate}
\item Short edges: For every $x \in \omega$, $(x,x+1) \in \tilde{E}$.
\item Long edges:
\begin{itemize}
\item If $x$ is even then $(x,(x+3,I)) \in \tilde{E}$, and if $x$ is odd $(x,(x+3,II)) \in \tilde{E}$.
\item For every $x\in \omega$, $(x,I)$ is connected to $x$ and to $(x+3,I)$. Similarly for player $II$.
\end{itemize}
\end{enumerate}
\end{itemize}
\item The game $\tilde{G}$ has bounded out degree.
\item For $x,y \in \omega$, there is an edge in $G$ from $x$ to $y$ if and only if there is a path in $\tilde{G}$ from $x$ to $y$ whose only $G$-vertices are $x$ and $y$.
\item Given $x \in \omega$, a player has a winning strategy from $x$ in $G$ if and only if she has a winning strategy from $x$ in $\tilde{G}$.
\end{itemize}
\end{exa}
To generalize this example, we introduce the following definition, that is inspired by \cite{BrutschT22} section 3.4:

\begin{defi} \label{RepDefGraph}
A graph $G=(V,E)$ is \textbf{represented by a graph} $\tilde{G}=(\tilde{V},\tilde{E})$ if:
\begin{enumerate}
\item $\tilde{V} \supseteq V$. The vertices of $\tilde{G}$ that are not in $G$ are called ``proper $\tilde{G}$-vertices,''  whereas the ones of $G$ are $G$-vertices.
\item There is an edge between $v_1$ and $v_2$ in $G$ if and only if there is a path between $v_1$ and $v_2$ in $\tilde{G}$ whose only $G$-vertices are $v_1$ and $v_2$.
\end{enumerate}
\end{defi}

\begin{restatable}{prop}{RepresentedIsWeakInterpStat} \label{RepresentedIsWeakInterp} If $G=(V,E)$ is represented by $\tilde{G}=(\tilde{V},\tilde{E})$, then $G$ is both $\MSO$ and weak-$\MSO$ definable in $(\tilde{G},V)$.
\end{restatable}
\begin{proof}
In $\tilde{G}$, let $R(x,y)$ if $\tilde{E}(x,y)$ and $y \notin V$. Let $R'(x,y)$ be the transitive and reflexive closure of $R$, which is both $\MSO$ and weak-$\MSO$ definable in $(\tilde{G},V)$. 
Since $G$ is represented by $\tilde{G}$, for every $x,y \in V$, $E(x,y)$ if and only if either $\tilde{E}(x,y)$, or there is $z$ such that $R'(x,z) \land \tilde{E}(z,y)$.
\end{proof}
\smallskip

\begin{defi} \label{RepDef}
A parity game $G$ is said to be \textbf{represented by a parity game} $\tilde{G}$ if:
\begin{enumerate}
\item $G$ is represented by $\tilde{G}$ as a graph.
\item The partition of the $G$-vertices between the 2 players is as in $G$.
\item A proper $\tilde{G}$-vertex is not reachable by proper $\tilde{G}$-vertices from both a $G$-vertex of player $I$ and a $G$-vertex of player $II$.
\item A proper $\tilde{G}$-vertex belongs to player $I$ if and only if it is reachable by proper $\tilde{G}$-vertices from a $G$-vertex of player $I$. Similarly for player $II$.
\item The colors of the $G$-vertices are the original colors + 2.
\item The proper $\tilde{G}$-vertices that belong to player $I$ have color $1$, and those of player $II$ have color $0$.
\end{enumerate}

\end{defi}

Note that if player $I$ stays from some point on her proper $\tilde{G}$-vertices, she loses, and similarly for player $II$.

The next proposition is immediate:
\begin{prop}
If a parity game $G$ is represented by a parity game $\tilde{G}$, then:
\begin{enumerate}
\item Given $x \in G$, a player has a winning strategy from $x$ in the game $G$ if and only if she has a winning strategy from $x$ in the game $\tilde{G}$. In other words, $winReg^{\tilde{G}}_I \cap G = winReg^{G}_I$, and similarly for player $II$.
\item Given $A \subseteq G$, a player has a uniform memoryless winning strategy from $A$ in the game $G$ if and only if she has a uniform memoryless winning strategy from $A$ in the game $\tilde{G}$.
\end{enumerate}
\end{prop}

Proposition \ref{WinRepresentInduceWin} plays a major role and reduces winning strategies of a game to the winning strategies of its representation:

\begin{restatable}{prop}{WinRepresentInduceWinStat} \label{WinRepresentInduceWin}
If $G$ is represented by $\tilde{G}$ and $win^{\tilde{G}}_I(x,y)$, $win^{\tilde{G}}_{II}(x,y)$ are the uniform memoryless winning strategies of $\tilde{G}$, then the uniform memoryless strategies of $G$ are definable in $(\tilde{G},V,win^{\tilde{G}}_I, win^{\tilde{G}}_{II})$.
\end{restatable}
\begin{proof}
Let $winReg^{\tilde{G}}_I(x), winReg^{\tilde{G}}_{II}(x)$ be formulas for the winning regions of the players. Consider the edge relation $W(x,y)$ on $\tilde{G}$ given by the winning strategies, which is, $win^{\tilde{G}}_I$ on $winReg^{\tilde{G}}_I$ and $win^{\tilde{G}}_{II}$ on $winReg^{\tilde{G}}_{II}$. The graph $(\tilde{G},W)$ represents $G$ with its uniform memoryless winning strategies. The same arguments as in the proof of
  Proposition \ref{RepresentedIsWeakInterp} show that  these strategies are definable in $(\tilde{G},V,win^{\tilde{G}}_I, win^{\tilde{G}}_{II})$.
\end{proof}
\subsubsection{From graph representation to game representation} \hfill\\
We first set the following terminology for definable structures and graphs:
\begin{defi}
Let $\mathcal{M} \subseteq \mathcal{N}$ be definable in $\mathcal{N}$.
\begin{enumerate}
\item Let $G$ be a \textbf{graph} definable in $\mathcal{M}$. We say that $G$ is \textbf{represented in} $\mathcal{N}$ \textbf{by} $\tilde{G}$ if $\tilde{G}$ is a graph definable in $\mathcal{N}$ that represents $G$.
\item Let $G$ be a \textbf{parity game} definable in $\mathcal{M}$. We say that $G$ is \textbf{represented in} $\mathcal{N}$ \textbf{by} $\tilde{G}$ if $\tilde{G}$ is a parity game definable in $\mathcal{N}$ that represents $G$.

\end{enumerate}
\end{defi}

Given a parity game $G$, we will say it is represented as a graph if the graph structure of $G$ is represented, and represented as a game if the game $G$ is represented. If $G$ is represented as a graph by a graph $\tilde{G}$, it might be impossible to extend this representation to a game representation, as a path in $\tilde{G}$ that represents an edge of player $I$ might intersect another path that represents an edge of player $II$ (this will violate condition (4) of definition \ref{RepDef}). However, by taking two copies of $\tilde{G}$ we can make these paths disjoint. That is the content of the following definition and lemma:
\begin{defi} 
Let $G=(V_1,V_2,E,C_0,\dots,C_{d-1})$ be a parity game. Let $\pi_0,\pi_1$ be the isomorphic embeddings $G \to G\times\{0\}, G \to G\times \{1\}$, respectively, and let $\pi$ be the embedding $\pi_0 \restriction {V_1} \cup \ \pi_1 \restriction {V_2}$.

The \textbf{splitted-arena game} $G_2$ is a game with arena $G \times \{0,1\}$ in which the vertices of player $I$ are $\pi_0(V_1)$ and the vertices of player $II$ are $\pi_1(V_2)$. In $G_2$, there is an edge from $x$ to  $y$   if there is an edge from $\pi^{-1}(x)$ to $\pi^{-1}(y)$   in $G$. Similarly, in $G_2$, $x$ has color $i$ if $\pi^{-1}(x)$ has color $i$.
\end{defi}

\begin{restatable}{lem}{GraphRepToGameRepStat} \label{GraphRepToGameRep} Let  $G$ be an $\mathcal{N}$-definable parity game. If $G$ is represented as a graph in $\mathcal{N}$ by a graph $\tilde{G}$, then there is an $\mathcal{N} \times 2$-definable parity game $\tilde{G_2}$ such that the splitted-arena game $G_2$ is represented as a game in $\mathcal{N} \times 2$ by $\tilde{G_2}$, and:
\begin{enumerate}
\item If $\tilde{G}$ has bounded (out) degree, then so does $\tilde{G_2}$.
\item If $\tilde{G}$ has an $\mathcal{N}$-definable local linear order, then $\tilde{G_2}$ has an $\mathcal{N}\times 2$-definable local linear order.
\item If $\tilde{G_2}$ has $\mathcal{N} \times 2$-definable uniform memoryless winning strategies, then $G$ has $\mathcal{N}$-definable uniform memoryless winning strategies.

Moreover, the $\mathcal{N}$-definable  uniform memoryless winning strategies are computable from the $\mathcal{N} \times 2$-definable uniform memoryless winning strategies.
\end{enumerate}
\end{restatable}
\begin{proof}
Let $G=(V_1,V_2,\Pi_0,\dots,\Pi_{d-1},E)$ be a parity game that is  represented as a graph by the $\mathcal{N}$-definable graph $\tilde{G}$. We define a game $\tilde{G_2}$ in $\mathcal{N} \times 2$ that represents $G_2$.

The vertices of the game $\tilde{G_2}$ are \[\big(V_1 \times \{0\}\big) \cup \big(V_2 \times \{1\}\big) \cup \big( (\tilde{V} \setminus V) \times \{0,1\} \big).\] The partition between the two players is according to the last coordinate.

A pair of vertices $((m_1,j_1), (m_2,j_2))$ is an edge of $\tilde{G_2}$ if $(m_1,m_2)$ is an edge of $\tilde{G}$ and either $j_1=j_2$, or $m_2 \in V$.

The colors are as expected of a parity game representation.

It is straightforward to show that $\tilde{G_2}$ represents $G_2$. $(1)$ is immediate, and $(2)$ is straightforward as well. $(3)$ follows from Lemma \ref{copyingLem} and Proposition \ref{WinRepresentInduceWin}.
\end{proof}

\smallskip

\subsubsection{Copy representations} \hfill\\
We can now introduce the main form of representation we will use in this paper - representation in a copying of a structure within a class of graphs / games.

Given a copying of $\mathcal{M}$, $\mathcal{M} \times k$, and $0 \leq i \leq k-1$, let $\pi_i: \mathcal{M} \to \mathcal{M} \times k$ be the embedding given by $\pi_i(m)=(m,i)$.

In most cases, when representing a graph $G$ by another graph $\tilde G$, we will be interested both in the structure $\mathcal M$ in which $G$ is defined, and in specific properties of $\tilde G$, such as having a bounded degree. This is why the definition below captures both the defining structure and the class of the representing graph:

\begin{defi} Let $G$ be a graph / game definable in $\mathcal{M}$, and $\mathcal{C}$ a class of graphs / games.

$G$ is \textbf{copy represented in} $\mathcal{M}$ \textbf{within} $\mathcal{C}$ if for some $k$ and $i$, there is an $\mathcal{M}\times k$-definable graph / game $\tilde{G} \in \mathcal{C}$ such that $\pi_i(G)$ is represented in $\mathcal{M} \times k$ by $\tilde{G}$ .

When $G$ is copy represented in $\mathcal{M}$ within $\mathcal{C}$, we may assume that for some $k$, $\pi_0(G)$ is represented in $\mathcal{M} \times k$ by $\tilde{G} \in \mathcal{C}$.
\end{defi}
For example, if $\mathcal{C}$ is the class of bounded degree graphs,
then $G$ is copy represented in $\mathcal{M}$ within $\mathcal{C}$ if there is an $\mathcal{M} \times k$-definable bounded degree graph $\tilde{G}$ that represents $G$, or more precisely, represents $\pi_0(G)$.

\subsection{Definable strategies in unbounded degree games} \label{MainSchemeSub}

We can now prove the correctness of our scheme for finding definable strategies in unbounded degree games:
\begin{restatable}{thm}{copyingDefStrategyStat} \label{copyingDefStrategy}
If $\mathcal{M}$ has the \textbf{selection property}, and $G$ is an $\mathcal{M}$-definable parity game that is copy represented in $\mathcal{M}$ as a graph by a graph that has \textbf{bounded out degree} and a \textbf{definable local linear order}, then:
\begin{enumerate}
\item There are $\mathcal{M}$\textbf{-definable uniform memoryless winning strategies} for $G$.
\item If selection for $\mathcal{M}$ is solvable, then the winning strategies are computable.
\end{enumerate}
\end{restatable}
\begin{proof} Let $k,i \in \omega$ and $\tilde{G}$ an $\mathcal{M} \times k$-definable graph that has bounded out degree and a definable local linear order such that $\pi_i(G)$ is represented in $\mathcal{M}\times k$ by $\tilde{G}$. By Lemma \ref{GraphRepToGameRep}, $\big(\pi_i(G)\big)_2$ is represented as a game in $\mathcal{M} \times k \times 2$ by a game $\tilde{G}_2$ that has bounded out degree and a definable local linear order.

By Lemma \ref{copyingSelection}, $\mathcal{M} \times k \times 2$ has the selection property as well, and has solvable selection if $\mathcal{M}$ has solvable selection. It then follows from Theorem \ref{boundedOutDegree} that there are $\mathcal{M} \times k \times 2$-definable uniform memoryless winning strategies for $\tilde{G}_2$. Using Lemma \ref{GraphRepToGameRep} again, there are $\mathcal{M} \times k$-definable uniform memoryless winning strategies for $\pi_i(G)$. Lemma \ref{copyingLem} concludes the proof.
\end{proof}

\smallskip

Due to Lemma \ref{DefLocLinearOrder}:
\begin{cor} \label{copyingDefStrategyBdd}
If $\mathcal{M}$ has the \textbf{selection property}, and $G$ is an $\mathcal{M}$-definable parity game that is copy represented in $\mathcal{M}$ as a graph by a graph that has \textbf{bounded degree}, then there are $\mathcal{M}$\textbf{-definable uniform memoryless winning strategies} for $G$.

Moreover, if selection for $\mathcal{M}$ is solvable, then the winning strategies are computable.
\end{cor}

\subsection{Definability  by regular expressions and representations} \label{RegExpSub}

When a graph $G$ represents a graph $H$, paths of $G$ correspond to edges of $H$. Thus, a standard way of describing a set of paths of $G$ helps in constructing a graph represented by $G$. Whenever $G$ is a labelled graph, regular expressions over its labels are a standard way of describing sets of paths. If in addition $G$ has bounded degree, the represented graphs are copy represented in $G$ by a bounded degree graph, and that might be the first step in finding a definable strategy. That motivates the following definitions:

Let $\Sigma$ be a finite alphabet $\Sigma = \{\sigma_1,\dots,\sigma_k\}$.
\begin{defi} \label{DefLabelledGraphWithLanguage}
\begin{enumerate}
\item A $\Sigma$\textbf{-labelled graph} is a graph whose vertices and edges have labels in $\Sigma$. As a structure, it has monadic predicates $V_1,\dots,V_k$ for the vertices labelled by $\sigma_1,\dots,\sigma_k$, respectively, and binary predicates $E_1,\dots,E_k$ for the edges labelled by $\sigma_1,\dots,\sigma_k$, respectively.
\item Given a path in a labelled graph, a \textbf{label of the path} is a concatenation of labels of the vertices and edges that form the path. More precisely, if $(v_1,\dots,v_n)$ is a path, and for every $1 \leq i \leq n-1$, $e_i$ is the edge $(v_i,v_{i+1})$, then a label of the path is a concatenation of labels of the sequence $(v_1,e_1,\dots,v_{n-1},e_{n-1},v_n)$.

In particular, the label of a single vertex path is a label of that vertex.

The label of the empty path is the empty string,
\item Let $L \subseteq \Sigma^*$ be a language over $\Sigma$, and $G$ a $\Sigma$-labelled graph. $L$ defines an edge relation $E_L$ on $G$ by: $(u,v)\in E_L$ if there is a path from $u$ to $v$ and a label of that path in $L$.
\end{enumerate}
\end{defi}
\begin{rem}
The label doesn't have to be unique - a vertex or an edge may have multiple labels, in which case a path has multiple labels. It is enough for the language $L$ to accept some label of a path in order for that path to induce an edge of $E_L$.

\end{rem}

\begin{defi}
A graph $H=(V,E)$ is \textbf{definable by regular expressions} in a $\Sigma$-labelled graph $G$ if $V$ is $\MSO$-definable in $G$, and there is a regular language $L$ over $\Sigma$ such that $E=E_L$.

\end{defi}

Our notion of definability by regular expressions is closely related to
Caucal’s notion of rational inverse mapping, which plays a central role
in his work on the pushdown hierarchy \cite{Caucal03}

The following example closely follows the notation of Definition \ref{DefLabelledGraphWithLanguage}:

\begin{exa}
Let $\Sigma = \{\sigma_1,\sigma_2,\sigma_3\}$ where $\sigma_1=0$, $\sigma_2=1$, $\sigma_3=-1$. Let $G=(V_1,V_2,V_3,E_1,E_2,E_3)$ be the following $\Sigma$-labelled graph on $\omega$:

\indent $V_1=\omega. $ \\
\indent $ E_2 = (n,n+1).$ \\ 
\indent $ E_3 = (n,n-1).$ \\\indent $ V_2=V_3=E_1 = \emptyset.$ \\
Let $H$ be a graph on $\omega$ whose edges are \\
\indent $\{(n,n+3\cdot m)\ |\ n,m\in\omega\}$ \\
and \\
\indent $\{(n,n-3\cdot m)\ |\ n,m\in\omega;\ n-3\cdot m \in \omega\}$.

\begin{itemize}
\item $H$ is definable by regular expressions in $G$: the regular expression \[\Big(\big((0\cdot 1)^3\big)^* + \big((0\cdot -1)^3\big)^*\Big)\cdot 0\] defines the edge relation.
\item $H$ is MSO-definable in $G$.
\item $G$ has bounded out degree, and $H$ has unbounded out degree.
\end{itemize}
We now define in a copying of $\omega$ a graph $\tilde{H}$ that represents $H$. Consider \[\tilde{V} = \omega \cup (\omega \times \{1,-1\} \times \{0,1,2\}).\] We denote $(x,y) \in \tilde{E}$ by $x\to y$:
\begin{enumerate}
\item $n \to (n+1,1,1)$ and $n \to (n-1,-1,2)$.
\item $(n,1,m) \to (n+1,1,m+1(mod\ 3))$, and $(n,-1,m)\to (n-1,-1,m-1(mod\ 3))$.
\item $(n,x,0) \to n$.
\end{enumerate}
The resulting graph $\tilde{H}$ demonstrates that $H$ is copy represented in $\omega$ by a bounded out degree graph.
\end{exa}
\smallskip

The following is a well-known folklore. For completeness, we give here the proof:
\begin{restatable}{prop}{RegExpWeakDefStatement} \label{RegExpWeakDef} Let $G$ be a $\Sigma$-labelled graph, and let $H=(V,E)$ be a graph definable by regular expressions in $G$. Then $H$ is both $\MSO$ and weak-$\MSO$ definable in $(G,V)$.
\end{restatable}
\begin{proof}
Let $\mathcal{A}$ be a deterministic automaton over $\Sigma$ such that $E_{\mathcal{L}(\mathcal{A})}=E_H$, and let $Q^\mathcal{A}=\{q_0,\dots,q_{n-1}\}$, where $q_0$ is the initial state. We may assume $q_0$ has no incoming edges.

The idea is as follows: for $v_0 \in X \subseteq G$, we will add predicates $\bq(v_0,X)=Q_1,\dots,Q_{n-1}$ for any non-initial state of $Q^\mathcal{A}$. For $v\in X$, $Q_i(v)$ will mean that there is a path in $X$ from $v_0$ to $v$ with a label on which the run of $\mathcal{A}$ ends in the state $q_i$.

Let then $\psi(v_0,X,\bq)$ be a formula saying that $v_0 \in X$ and $\bq \subseteq X$ is minimal such that $\bq$ is closed under:
\begin{itemize}
\item If $l$ is a label of $v_0$ and the run of $\mathcal{A}$ on $l$ ends in $q_i$, then $Q_i(v_0)$.
\item If:
\begin{itemize}
\item $v',v\in X$, $Q_i(v')$, $E(v',v)$, and $l_1,l_2$ are some labels of the edge from $v'$ to $v$, and of $v$, respectively.
\item The run of $\mathcal{A}$ from $q_i$ on $l_1\cdot l_2$ ends in $q_j$.
\end{itemize}
then $Q_j(v)$.
\end{itemize}

The following is a (weak-)$\MSO$ definition of $E(v_0,v)$: there exist (finite) $X,\bq$ such that \[\psi(v_0,X,\bq) \land \bigvee_{q_i\ accepting}Q_i(v).\]
\end{proof}

\smallskip
When we use representation for the purpose of finding definable strategies, we would want to represent a game of unbounded degree by a game of bounded degree. We've seen that an unbounded degree game $H$ might be definable by regular expressions in a graph $G$ of bounded degree. Whenever that happens, we can find a bounded degree representing graph that is definable in $G \times k$:

\begin{restatable}{thm}{RegExpToReprStat} \label{RegExpToRepr}
Let $G=(V,E)$ be a bounded (out) degree $\Sigma$-labelled graph, and let $H$ be a graph definable by regular expressions in $G$. Then $H$ is effectively copy represented in $G$ by a bounded (out) degree graph.

Moreover, if $G$ has definable local linear order, then so does the representing graph.
\end{restatable}
\begin{proof} Let $\mathcal{A}$ be a deterministic automaton over $\Sigma$ such that $E_{\mathcal{L}(\mathcal{A})}=E_H$, and let $q_0 \in Q^\mathcal{A}$ be the initial state. We may assume $q_0$ has no incoming edges.

Let $V_{\tilde{H}}$ be the set $V \times Q^\mathcal{A}$. As to the edges:
\begin{itemize}
\item For every $v \in V$ and $q \in Q^\mathcal{A}$, there is an edge from $(v,q_0)$ to $(v,q)$ if in $\mathcal{A}$, a label of the vertex $v$ leads from $q_0$ to $q$.
\item For every $v_1,v_2 \in V$ and $q_1,q_2 \neq q_0$, there is an edge from $(v_1,q_1)$ to $(v_2,q_2)$ if $E_G(v_1,v_2)$, and a label of that edge followed by a label of $v_2$ leads in $\mathcal{A}$ from $q_1$ to $q_2$.
\item For every $v \in V$ and $q \in Q^\mathcal{A}$, there is an edge from $(v,q)$ to $(v,q_0)$ if $q$ is accepting.
\end{itemize}

It is left to convince ourselves that $\tilde{H}$ has bounded (out) degree, and that it represents $\pi_{q_0}(H)$.

As to the 'moreover' part, a local linear order on $G$ may easily be extended to $H$.
\end{proof}

The main result of this subsection follows:
\begin{cor} \label{MainCorRegExp}
Let $G=(V,E)$ be a $\Sigma$-labelled graph that has the \textbf{selection property}, and let $H$ be a game whose graph is \textbf{definable by regular expressions} in $G$. If either:
\begin{itemize}
    \item $G$ has \textbf{bounded degree},
    
    or
    
    \item $G$ has \textbf{bounded out degree} and a \textbf{definable local linear order},
\end{itemize}then $H$ has $G$\textbf{-definable uniform memoryless winning strategies}.

Moreover, if selection for $G$ is solvable, then the winning strategies are computable.
    
\end{cor}
\begin{proof}
This is a corollary of Theorem \ref{RegExpToRepr} together with either Theorem \ref{copyingDefStrategy} or Corollary \ref{copyingDefStrategyBdd}.
\end{proof}
\begin{rem}[Prefix-recognizable graphs]\label{rem:cachat}
In \cite{CachatWaluk}, Cachat considered prefix-recognizable graphs and games. It is easy to see
that these graphs are definable by  regular expressions in an expansion of the full binary tree $\mathcal{T}_2$ by MSO-definable
unary predicates (see also \cite{Caucal03} in which it is shown that prefix-recognizable graphs are exactly the graphs obtained
by applying inverse rational mappings to the full binary tree). By Corollary \ref{MainCorRegExp}, and
since $\mathcal{T}_2$ has the selection property,
prefix recognizable games have  uniform memoryless winning strategies MSO-definable
in $\mathcal{T}_2$. This implies the main result of \cite{CachatWaluk}.

\end{rem}
\subsection{Definability by regular expressions vs MSO-definability}\label{subsectMsovs-reg}
Proposition \ref{RegExpWeakDef} states that definability by   regular expressions implies
MSO-definability.
Theorems \ref{MainTheorem} and \ref{MainThmOrd} describe important cases in which
the reverse implication holds.

\begin{restatable}[$\MSO$-definability implies regular expressions definability]{thm}{MainThmNatStatement} \label{MainTheorem}
Let $\mathcal{M}$ be one of the following structures:
\begin{enumerate}
\item $(\omega,<,\bar{P})$.
\item \label{Clause2MainTheorem} $(T_k,\bar{P})$.
\end{enumerate}
Let $H$ be a graph $\MSO$-definable in $\mathcal{M}$. Then there exist computable $\Sigma$ and a \textbf{bounded degree} $\Sigma$-labelled graph $G$ $\MSO$-definable in $\mathcal{M}$ such that $H$ is effectively definable by regular expressions in $G$.

Moreover, if $M$ is the universe of $\mathcal M$, then the original structure $\mathcal{M}$ and the new structure $(M,G)$ are mutually definable, hence $(M,G)$ has the selection property if and only if $\mathcal{M}$ has the selection property.
\end{restatable}
\smallskip

\begin{rem}
A reviewer pointed out that Theorem \ref{MainTheorem} (\ref{Clause2MainTheorem}), which was obtained independently
in the present work, has an earlier counterpart in Carayol
\cite[Proposition~3.2.1]{CarayolThesis}. Carayol's result is formulated in
terms of inverse rational substitutions; the equivalence between that
formulation and our regular-expression formulation is standard. Thus,
although the present formulation and proof were developed independently, the
substance of \ref{MainTheorem} (\ref{Clause2MainTheorem}) was already established in
\cite[Proposition~3.2.1]{CarayolThesis}.

The result was also stated in English in
\cite[Proposition~1]{BraudCarayol}. We note, however, that the quantifier
order in \cite[Proposition~1]{BraudCarayol}, as printed, differs from that
of the uniform statement proved in
\cite[Proposition~3.2.1]{CarayolThesis} and used here.
\end{rem}

\begin{proof}[Proof of theorem \ref{MainTheorem}] Let $H$ be a graph definable in $\mathcal{M}$ by formulas $V(x), E(x,y)$. We will find $\Sigma$ and an $\mathcal{M}$-definable bounded degree $\Sigma$-labelled graph $G$ such that $H$ is definable by regular expressions in $G$.

The proof is divided into cases according to the structure $\mathcal{M}$. In both cases the proof is based on the composition method (see subsections \ref{TypesSub} and \ref{CompSub}).

\begin{enumerate}

\item $\mathcal{M}=(\omega,<,\bp)$:

The construction of $G$ is based on Lemma \ref{BinaryOrdinalParam}. Let $n \in \omega$ be such that types of rank $n$ determine $E(x,y)$, as in the lemma, and let $l = |\bp|$. Let $\mathcal{T}^n_l$ be the set of types of rank $n$ over $l$ free variables.

We define the following labelled graph $G$ on $\omega$:
\begin{itemize}
\item The vertices are $\omega$. The label of a vertex $m \in \omega$ is a pair of types from $\mathcal{T}^n_l$ - a prefix label \[tp^n([0,m),<,\bp)\] and a suffix label \[tp^n([m,\infty),<,\bp).\]
\item The edges have both a sign label in $\{0,+1,-1\}$ and a type label in $\mathcal{T}^n_l$. For every $m$, $(m,m+1)$ is an edge with sign label $+1$, and $(m,m)$ is an edge with sign label $0$. For every $m>0$, $(m,m-1)$ is an edge with sign label $-1$. All of these edges have the same type label: $tp^n(\{m\},<,\bp)$; this type $tp^n(\{m\},<,\bp)$
is determined by the tuple of booleans $m\in P_i$.
\end{itemize}

The labelled graph $G$ is definable in $(\omega,<,\bp)$, and has a bounded degree. Also, $(\omega,<,\bp)$ is definable in $G$, hence so does $H$. It is left to find a regular language $L_H$ such that $E_{L_H} = E_H$.

\begin{clm} \label{RegSumTypes}
Let $\tau \in \mathcal{T}^n_l$. The language over $\mathcal{T}^n_l$ of types whose total sum is $\tau$ is regular
\end{clm}
\begin{proof}
The addition of types in $\mathcal{T}^n_l$, as defined in Definition \ref{AdditionOfTypes}, is a well defined operation on a finite set.
\end{proof}
We use this claim to find regular languages $L_{\self}, L_{\forward},L_{\backward}$ whose union $L_H$ satisfies $E_{L_H} = E_H$.

The language $L_{\self}$ accepts a label of a single edge path from a vertex to itself if the prefix and suffix labels of the vertex determine a self $E$-edge.

The language $L_{\forward}$ accepts a label of a path if there are types $t_1,t_2,t_3$ such that:
\begin{enumerate}
\item The following conditions imply $E(x,y)$:
\begin{itemize}
\item $x < y.$
\item $tp^n([0,x),<,\bp)=t_1.$
\item $tp^n([x,y),<,\bp)=t_2.$
\item $tp^n([y,\infty),<,\bp)=t_3.$
\end{itemize}
\item The path begins in a vertex with prefix label $t_1$, followed by positive edges whose sum of type labels is $t_2$, and ends in a vertex of suffix label $t_3$.
\end{enumerate}

The language $L_{\backward}$ accepts a label of a path if there are types $t_1,t_2,t_3$ such that:
\begin{enumerate}
\item The following conditions imply $E(x,y)$:
\begin{itemize}
\item $y<x.$
\item $tp^n([0,y),<,\bp)=t_1.$
\item $tp^n([y,x),<,\bp)=t_2.$
\item $tp^n([x,\infty),<,\bp)=t_3.$
\end{itemize}
\item The path begins in a vertex with suffix label $t_3$, followed by negative edges whose sum of type labels is $t_2$, and ends in a vertex of prefix label $t_1$.\footnote{It is interesting to note that when $E$ is symmetric, $L_{backward}$ is the mirror language of $L_{forward}$.}
\end{enumerate}

By Claim \ref{RegSumTypes}, $L_{forward}$ and $L_{backward}$ are regular, whereas $L_{self}$ is evidently regular. By their definitions, $E_{L_{self}}$ are the self edges, $E_{L_{forward}}$ are the forward edges, and $E_{L_{backward}}$ are the backward edges. Therefore for $L_H=L_{self}\cup L_{forward} \cup L_{backward}$, $E_{L_H}=E_H$.

\item $\mathcal{M}=(T_k,\bp)$:

The following lemma is used in the proof:
\begin{lem}[$T_k$ composition theorem for successors, Theorem 1.12 \cite{LifschesS98}] \label{TkCompSucc}  For any $d,l \in \omega$ there is a computable $n=n(d,l)$ such that for every $l$-tuple $\bp$ the multi-set of \[tp^n(T_{\geq 0},<,\bp),\dots, tp^n(T_{\geq k-1},<,\bp)\] effectively determines \[tp^d(\cup T_{\geq i},<,\bp).\]
\end{lem}
The construction is based on Theorem \ref{BinaryTkParam}. Let $n,d \in \omega$ be such that types of rank $n$ expanded by types of rank $d$ determine $E(x,y)$, as in the theorem.
Below, we also use the notations from Definition \ref{def:expansion-path}.

Let $G$ be the following labelled graph defined in $(T_k,\bp)$:
\begin{itemize}
\item The vertices are $T_k$. Every $z \in T_k$ has a prefix label \[tp^n([\epsilon,z),<,\bp,\bq^{d}_{\epsilon,z,\bp}),\] a self label \[tp^n(\{z\},<,\bp),\] a suffix label \[tp^n(T_{\geq z},<,\bp),\] and suffix labels of successors \[tp^n(T_{\geq z\cdot 0},<,\bp),\] \qquad \qquad\qquad\qquad \qquad\qquad\qquad\qquad\vdots \[tp^n(T_{\geq z\cdot (k-1)},<,\bp).\]
\item The edges of $G$ are either $(z, z \cdot i)$ with a sign label $+$ and a type label $i$, or $(z \cdot i, z)$ with a sign label $-$ and a type label $i$.
\end{itemize}

By Lemma \ref{TkCompSucc}, we may assume that for every $z$ and $i$, $tp^d(T_{>z\perp [z,z\cdot i)})$ is determined by the suffix labels of successors. It follows that from the label of a forward path on a path $[z,y]$, an automaton may calculate $tp^n([z,y),<,\bp,\bq^{d}_{z,y,\bp})$, and similarly for a backward path. Given a type $t$, let $L_t^{\forward}$ be the regular language of labels of a forward path whose accumulative $d$-expanded type is $t$. The regular language $L_t^{\backward}$ is defined in a similar way.

The labelled graph $G$ is definable in $(T_k,\bp)$, and has a bounded degree. Since the self label of $z$ determines $\bp(z)$, $(T_k,\bp)$ is also definable in $G$, and so it is left to find a regular language $L$ that defines $E_H$ by regular expressions.

This language $L$ accepts a path if there are types $t_1,t_2,t_3,t_4,t_5,s_1,\dots,s_{k-2}$ such that:
\begin{enumerate}
\item Whenever $z$ is the maximal common prefix of $x$ and $y$, the following conditions imply $E(x,y)$:
\begin{itemize}
\item $tp^n([\epsilon,z),<,\bp,\bq^{d}_{\epsilon,z,\bp})=t_1$.

\item $tp^n([z,x),<,\bp,\bq^{d}_{z,x,\bp})=t_2.$
\item $tp^n(T_{\geq x},<,\bar{P})=t_3.$
\item $tp^n([z,y),<,\bp,\bq^{d}_{z,y,\bp})=t_4.$
\item $tp^n(T_{\geq y},<,\bar{P})=t_5.$
\item The multi-set of \[\{tp^n(T_{\geq z \cdot j},<,\bar{P})\ |\ (x \not \geq z \cdot j) \land (y \not \geq z \cdot j)\}\] is the multi-set of $\{s_1,\dots,s_{k-2}\}$.

\end{itemize}

\item The path begins in a vertex with suffix label $t_3$, followed by a backward path whose label is in $L^{\backward}_{t_2}$, then a vertex with prefix label $t_1$ and suffix labels of successors $s_1,\dots,s_{k-2}$, followed by a forward path whose label is in $L^{\forward}_{t_4}$, and ends in a vertex with suffix label $t_5$.
\end{enumerate}
We claim that $L$ is as sought. Indeed, since $n,d\in\omega$ were chosen such that the types determine $E(x,y)$ (as in Theorem \ref{BinaryTkParam}), and the language $L$ accepts exactly the paths coding these types, the conclusion follows.

\end{enumerate}
\end{proof}

Note that the case of $\omega$ without parameters has an elementary proof that is based on the automaton representing the edge relation of the graph. 
The case of $\mathcal{T}_2$ without parameters can  be
extracted from L\"{a}uchli  and  Savioz  \cite{LS87}.

\smallskip

In order to extend this result to any ordinal below $\omega^\omega$, we first make the following observation:
\begin{lem} \label{BddDegAlpha}
Let $\alpha < \omega^\omega$. There is a binary relation $E_\alpha$ on $\alpha$ such that $(\alpha,E_\alpha)$ is a bounded out degree graph, and $(\alpha,<)$, $(\alpha,E_\alpha)$ are mutually definable.
\end{lem}
\begin{proof}
Let \[\alpha = \omega^{m-1} \cdot k_{m-1} + \omega^{m-2} \cdot k_{m-2} + \dots + \omega \cdot k_1 + k_0\] be the Cantor normal form of $\alpha$.

The $E_\alpha$ edges are defined as follows: for every $\delta < \alpha$ and every $\beta \in \{1,\omega,\omega^2,\dots,\omega^{m-1}\}$, if $\delta+\beta<\alpha$ then there is an $E_\alpha$ edge $(\delta,\delta+\beta)$. This relation is definable in $(\alpha,<)$, and $<$ is definable in $(\alpha,E_\alpha)$.
\end{proof}
Note that the graph $(\alpha,E_\alpha)$ has
infinite in-degree -- for every limit node $\beta$
there are infinitely many edges that enter $\beta$.

\begin{restatable}{thm}{MainThmOrdStatement} \label{MainThmOrd}
Let $\alpha < \omega^\omega$ and $\bp \subseteq \alpha$. Let $H$ be a graph definable in $(\alpha,<,\bp)$. Then there exist computable $\Sigma$ and a \textbf{bounded out degree} $\Sigma$-labelled graph $G$ definable in $(\alpha,<,\bp)$ such that $H$ is effectively definable by regular expressions in $G$.

Moreover, $(\alpha,<,\bp)$ and $(\alpha,G)$ are mutually definable, hence $(\alpha,G)$ has the selection property.
\end{restatable}

\begin{proof} The construction is based on Theorem \ref{BinaryOrdinalParam}.

Let $H$ be a graph defined in $\alpha$, and let \[\alpha = \omega^{m-1} \cdot k_{m-1} + \omega^{m-2} \cdot k_{m-2} + \dots + \omega \cdot k_1 + k_0\] be the Cantor normal form of $\alpha$.

Let $n \in \omega$ be such that types of rank $n$ determine $E(x,y)$, as in Theorem \ref{BinaryOrdinalParam}. Let $l=|\bp|$, and let $\mathcal{T}^n_l$ be the set of types over $\mathcal{L}=\{<\}$ of rank $n$ with $l$ free variables.

Let $G$ be the following labelled graph defined on $\alpha$:
\begin{itemize}
\item The label of a vertex $\delta \in \alpha$ is a tuple of types from $\mathcal{T}^n_l$: a prefix label \[tp^n([0,\delta),<,\bp),\] a suffix label \[tp^n([\delta,\infty),<,\bp),\] and intermediate labels \[tp^n([\delta,\delta+\omega),<,\bp),\] \[tp^n([\delta,\delta+\omega^2),<,\bp),\]  \qquad\qquad\qquad\qquad\qquad \qquad\qquad\qquad\qquad\vdots \[tp^n([\delta,\delta+\omega^{m-1}),<,\bp).\]
\item The edges of positive label are $E_\alpha$ of Lemma \ref{BddDegAlpha}.

Note that for every $\delta_1<\delta_2$ there is a path of positive edges from $\delta_1$ to $\delta_2$.
\item As to the negative edges, for every $k\in\omega$ and $0 \leq l < m$, if $\omega^l \cdot (k+1) < \alpha$ then there is an edge \[(\omega^l\cdot (k+1) , \omega^l \cdot k).\] In particular, every successor ordinal is connected to its predecessor.

Note that for every $\delta_1<\delta_2$ there is a path from $\delta_2$ to $\delta_1$ composed of consecutive negative edges followed by consecutive positive edges.
\item The type label of an edge $(\delta_1,\delta_2)$ is \[tp^n([min(\delta_1,\delta_2),max(\delta_1,\delta_2)),<,\bp).\]
\item For every $\delta<\alpha$, there is a self edge of sign label $0$ (and no type label).
\end{itemize}

It is straightforward to check that $G$ is definable in $(\alpha,<,\bp)$ and has bounded out degree. Since $<$ is definable from $E_\alpha$ and $\bp$ is definable from the suffix labels, $H$ is definable in $G$. To complete the proof, we will find regular languages $L_{\self}, L_{\forward},L_{\backward}$ whose union $L_H$ satisfies $E_{L_H} = E_H$.

The language $L_{\self}$ accepts a label of a single edge path from a vertex to itself if the prefix and suffix labels of the vertex determine a self $E$-edge.

The language $L_{\forward}$ accepts a label of a path if there are types $t_1,t_2,t_3$ such that:

\begin{enumerate}
\item The following conditions imply $E(\delta_1,\delta_2)$:
\begin{itemize}
\item $\delta_1 < \delta_2.$
\item $tp^n([0,\delta_1),<,\bp)=t_1.$
\item $tp^n([\delta_1,\delta_2),<,\bp)=t_2.$
\item $tp^n([\delta_2,\infty),<,\bp)=t_3.$
\end{itemize}
\item The path begins in a vertex with prefix label $t_1$, followed by positive edges whose sum of type labels is $t_2$, and ends in a vertex of suffix label $t_3$.
\end{enumerate}

To construct the language $L_{\backward}$, first note that the type label of an edge determines the order type of the edge, so we may add to each edge an ordinal label that is given by its order type.

\medskip
Let us first take advantage of the ordinal labels to examine a specific example:

\begin{exa}
Let \[\delta_1=\omega^2\cdot 3 + \omega \cdot 4 +5\] and \[\delta_2 = \omega^2 + \omega \cdot 2 + 6.\] The edges of the path from $\delta_1$ to $\delta_2$ have the following ordinal and sign labels: \[(-1)^5 \cdot (-\omega)^4 \cdot (-\omega^2)^2 \cdot (+\omega)^2 \cdot (+1)^6.\] Whether there is an edge $E(\delta_1,\delta_2)$ is determined by:
\begin{enumerate}
\item The prefix label $tp^n([0,\delta_2),<,\bp)$.
\item The intermediate label $tp^n([\delta_2,\delta_2+\omega^2),<,\bp)=tp^n([\delta_2,\omega^2 \cdot 2),<,\bp)$.
\item The type $tp^n([\omega^2 \cdot 2,\delta_1),<,\bp)$, given by the  sum of type labels of the path edges that correspond to \[(-1)^5 \cdot (-\omega)^4 \cdot (-\omega^2).\]
\item The suffix label $tp^n([\delta_1,\infty),<,\bp)$.
\end{enumerate}

Indeed, the sum of $2)$ and $3)$ is equal to $tp^n([\delta_2,\delta_1),<,\bp)$; together with $1)$ and $4)$ this determines whether there is an edge in $G$ from $\delta_1$ to $\delta_2$.

Note that the type labels of the positive path edges were not explicitly taken into account in this computation, nor did the type label of the last negative path edge.

If, as a different example, $\delta_2$ will be replaced by ${\delta_2'}=\omega^2$, the path will only have negative edges. In that case, the intermediate label will no longer be needed, and instead the sum of type labels will be of all of the path edges.
\end{exa}
\medskip
Back to the construction of $L_{\backward}$, let $\tilde{L}_{\backward}$ be a language over edge labels that accepts words of the ordinal label form \[(-1)^* \cdot (-\omega)^* \cdots (-\omega^{d-1})^+ \cdot (+\omega^{d-2})^* \cdots (+\omega)^* (+1)^*,\] which is, words that begin with negative edges and continue with positive edges in such a way that the maximum power of the negative ordinal labels is strictly greater than the maximum power of the positive ordinal labels.

Let $t \in \mathcal{T}^n_l$. Let $L_t^{\backward}$ be the words of $\tilde{L}_{\backward}$ in which \textbf{there are no positive edges}, and the sum of the type labels of all edges is $t$. Let $L_t^{mixed}$ be the words of $\tilde{L}_{\backward}$ in which \textbf{there are positive edges}, and the sum of the type labels of the \textbf{negative edges not including the last negative edge} is $t$.

The language $L_{\backward}$ then accepts a label of a path if one of the following is satisfied:
\begin{enumerate}

\item There are types $t_1,t_2,t_3$ such that:

\begin{enumerate}
\item The following conditions imply $E(\delta_1,\delta_2)$:
\begin{itemize}
\item $\delta_2 < \delta_1.$
\item $tp^n([0,\delta_2),<,\bp)=t_1.$
\item $tp^n([\delta_2,\delta_1),<,\bp)=t_2.$
\item $tp^n([\delta_1,\infty),<,\bp)=t_3.$
\end{itemize}
\item The path begins in a vertex with suffix label $t_3$, followed by a path whose edges have a label in $L_{t_2}^{\backward}$, and ends in a vertex of prefix label $t_1$.
\end{enumerate}

\item There are types $t_1,t_2,t_3,t_4$ such that:

\begin{enumerate}
\item The following conditions imply $E(\delta_1,\delta_2)$:
\begin{itemize}
\item $\delta_2 < \delta < \delta_1.$
\item $tp^n([0,\delta_2),<,\bp)=t_1.$
\item $tp^n([\delta_2,\delta),<,\bp)=t_2.$
\item $tp^n([\delta,\delta_1),<,\bp)=t_3.$
\item $tp^n([\delta_1,\infty),<,\bp)=t_4.$
\end{itemize}
\item The path begins in a vertex with suffix label $t_4$, followed by a path whose edges have a label in $L_{t_3}^{mixed}$, and ends in a vertex of prefix label $t_1$ and corresponding intermediate label $t_2$. By corresponding intermediate label we refer to the label corresponding to the last backward edge of the path. 
\end{enumerate}

\end{enumerate}

It is left to show that $L_H$ is as needed.

\end{proof}

\smallskip

The main result of this section now follows out of Corollary \ref{MainCorRegExp}:
\begin{thm} \label{MainDefStrategy}
Let $\mathcal{M}$ be one of the following structures:
\begin{enumerate}
\item $(\omega,<,\bar{P})$.
\item $(\mathcal{T}_k,\bar{P})$ that has the selection property.
\item $(\alpha,<,\bp)$ for $\alpha < \omega^\omega$.
\end{enumerate}

A parity game definable in $\mathcal{M}$ has $\mathcal{M}$-definable uniform memoryless winning strategies.

Moreover, if selection for $\mathcal{M}$ is solvable then the winning strategies are computable.
\end{thm}

\medskip

Let $P \subseteq T_2$. We've seen that if $(\mathcal{T}_2,P)$ has the selection property, then games definable in $(\mathcal{T}_2,P)$ have definable uniform memoryless strategies. However, when the selection property fails for $(\mathcal{T}_2,P)$, definable uniform memoryless strategies may fail to exist:

\begin{prop}[\cite{carayol2010choice}] \label{FailureT2P}
There is $P \subseteq T_2$ and a parity game definable in $\MSO(\mathcal{T}_2,P)$ such that there is no definable uniform memoryless winning strategy for player $I$.
\end{prop}

Intuitively, the obstruction comes from the fact that a choice function is not
$\MSO$-definable in the full binary tree \cite{ShelahGurevich}. Indeed, one can define in $(T_2,P)$ a one-player game in
which the player wins precisely by reaching a node in $P$. Any winning memoryless strategy
must therefore choose, at each position from which a win is possible, a successor leading
towards $P$. From such a strategy one can easily define a choice function for $P$.

\section{Generalized Church Problem}\label{GCPSection}

Let $\Sigma_{in},\Sigma_{out}$ be sets of alphabets (not necessarily finite), and let $R \subseteq {\Sigma_{in}}^\omega \times {\Sigma_{out}}^\omega$. In the following section we investigate the Generalized Church Problem for $R$ given by a definable infinite state automaton, and implementations given by definable infinite state transducers.

The framework of this section captures several known extensions of Church synthesis.
In particular, it subsumes the synthesis results of \cite{Rabinovich07,Rabinovich12a} over finite alphabets,
as well as the extension to infinite alphabets studied in \cite{BRUTSCHLT17, BrutschT22}.

\begin{defi}
\begin{enumerate}
\item An operator $f:{\Sigma_{in}}^\omega \to {\Sigma_{out}}^\omega$ \textbf{implements} $R$ if for every $x\in{\Sigma_{in}}^\omega$, $(x,f(x)) \in R$.
\item An operator $f:{\Sigma_{in}}^\omega \to {\Sigma_{out}}^\omega$ is \textbf{causal} if for every $x \in {\Sigma_{in}}^\omega$ and $n \in \omega$, $f(x)(n)$ depends only on $x(0),\dots,x(n)$.
\end{enumerate}
\end{defi}

\begin{defi}  \label{DADef}
\begin{enumerate}
\item A \textbf{deterministic (parity) acceptor} (DA) over an alphabet $\Sigma$ is \[\mathcal{D} = (V,\Sigma,v_0,\delta_\mathcal{D},\Pi_0,\dots,\Pi_{d-1})\] where $V$ is a set of nodes, $\Sigma$ is an alphabet, $v_0 \in V$ is an initial node, $\delta_\mathcal{D}:V \times \Sigma \to V$ is a transition function (which is, a deterministic and complete transition relation), and the partition $\Pi_0,\dots,\Pi_{d-1}$ are colors of the nodes.
\item Let $\mathcal{D}$ be a DA and $x \in \Sigma^\omega$ an $\omega$-word over $\Sigma$. The \textbf{run} of $\mathcal{D}$ on $x$ is a sequence of nodes $(w_0,w_1,\dots)$ such that $w_0=v_0$ and for every $n$, \[w_{n+1}=\delta_{\mathcal{D}}(w_n,x(n)).\]
The run is \textbf{accepting} if the minimal node color that repeats infinitely many times is even. The $\omega$-word $x$ is accepted by $\mathcal{D}$ if the run of $\mathcal{D}$ on $x$ is accepting. The \textbf{language} of $\mathcal{D}$, $\mathcal{L}(\mathcal{D})$, is the set of $\omega$-words accepted by $\mathcal{D}$.
\item Let $\mathcal{D}$ be a DA over alphabet $\Sigmain \times \Sigmaout$. The \textbf{input-output relation defined by} $\mathcal{D}$,  $R_\mathcal{D}$, is the set of pairs of sequences $(x,y) \in {\Sigmain}^\omega \times {\Sigmaout}^\omega$ such that the $\omega$-word over $\Sigmain \times \Sigmaout$ whose every $n$'th letter is $(x(n),y(n))$ is an element of $\mathcal{L}(\mathcal{D})$.

\item \label{TransducerDef}A \textbf{transducer} over $\Sigmain$ and $\Sigmaout$ is \[\mathcal{T} = (V,\Sigmain,\Sigmaout,v_0,\delta_\mathcal{T})\] where $V$ is a set of nodes, $\Sigmain, \Sigmaout$ are input and output alphabets, $v_0 \in V$ is an initial node, and $\delta_\mathcal{T}:V \times \Sigmain \to V \times \Sigmaout$ is a transition and output function.

\item Let $\mathcal{T}$ be a transducer and $x \in {\Sigmain}^\omega$ an $\omega$-word over $\Sigmain$. The \textbf{run} and \textbf{output} of $\mathcal{T}$ on $x$ are a sequence of nodes $(w_0,w_1,\dots)$ and an $\omega$-word $(\tau_0,\tau_1,\dots)$ over $\Sigmaout$, respectively, such that $w_0=v_0$, and for every $n$, \[(w_{n+1},\tau_n)=\delta_{\mathcal{T}}(w_n,x(n)).\]

\item \textbf{The operator induced by} $\mathcal{T}$ is the causal operator $f_\mathcal{T}: {\Sigma_{in}}^\omega \to {\Sigma_{out}}^\omega$ that takes each $x \in {\Sigmain}^\omega$ to its output.

Given a relation $R \subseteq {\Sigma_{in}}^\omega \times {\Sigma_{out}}^\omega$, we say that $\mathcal{T}$ \textbf{implements} $R$ if $f_\mathcal{T}$ implements $R$.
\end{enumerate}
\end{defi}

Note that in the above definition there are no assumptions of finiteness, which is, neither the set of nodes nor the alphabet are assumed to be finite.

One can similarly define nondeterministic automata in this setting. However, in general such automata need not admit definable determinization. Determinism is essential for the corresponding game-theoretic applications, since strategies are induced by the unique runs of the automaton.

\subsection{Generalized Church Problem for deterministic acceptors over finite alphabets}
\begin{defi}
Let $\mathcal{M}$ be a structure. A deterministic acceptor / transducer over finite alphabets is \textbf{definable in }$\mathcal{M}$ if:
\begin{itemize}
\item $V$, $v_0$, $\Pi_0,\dots,\Pi_{d-1}$ are definable subsets of $\mathcal{M}$.
\item For every $\sigma \in \Sigma_{in}$, the relation $v'=\delta_{\mathcal{D}}(v,\sigma)$ is definable in $\mathcal{M}$.
\item For every $\sigma \in \Sigma_{in}, \tau \in \Sigma_{out}$, the relation $(v',\tau)=\delta_{\mathcal{T}}(v,\sigma)$ is definable in $\mathcal{M}$.
\end{itemize}

\end{defi}
\smallskip
Let $\mathcal{M}$ be a structure, and let $\Sigmain, \Sigmaout$ be finite alphabets.
\begin{prob}[Generalized Church Problem for $\mathcal{M}$-definable deterministic acceptors over finite alphabets]
Let $\mathcal{D}$ be an $\mathcal{M}$-definable deterministic acceptor over $\Sigmain \times \Sigmaout$.

Is there a causal operator $f:{\Sigma_{in}}^\omega \to {\Sigma_{out}}^\omega$ that implements $R_\mathcal{D}$?

\end{prob}

\begin{thm} \label{gcpFinite} Let $\mathcal{D}$ be an $\mathcal{M}$-definable deterministic acceptor over $\Sigmain \times \Sigmaout$.
\begin{enumerate}
\item (Decidability) It is recursive in $MTh(\mathcal{M})$ whether there exists a causal operator that implements $R_\mathcal{D}$.
\item (Definability) If $\mathcal{M}$ has the selection property, and there is a causal operator that implements $R_\mathcal{D}$, then there is an $\mathcal{M}$-definable transducer over $\Sigmain$ and $\Sigmaout$ that implements $R_\mathcal{D}$.

Moreover, the vertices of the transducer are the vertices of $\mathcal{D}$, and if $\delta_\mathcal{T}(v,\sigma_{in})=(v',\sigma_{out})$ then $\delta_\mathcal{D}(v,\sigma_{in},\sigma_{out})=v'$.
\item (Computability) If selection for $\mathcal{M}$ is solvable, and there is a causal operator that implements $R_\mathcal{D}$, then it is computable to find a definition of a transducer that implements $R_\mathcal{D}$.
\end{enumerate}

\end{thm}
\begin{proof}
The proof follows similar steps as the ones taken by B\"uchi and Landweber \cite{BL69} in their solution of the Church problem. In the first step, the $DA$ is transformed into a parity game. In the second step, this game is solved, in the sense that a definable memoryless winning strategy is found. In the third step, this strategy is transformed into an implementing transducer. On each move of player output, the transducer will pick up a letter $\sigma_{out}\in\Sigmaout$ that is compatible with the defined strategy.
\begin{enumerate}
\item Let $G$ be the following parity game on the vertices \[(V \times \{orig\}) \cup (V \times \Sigmain):\] player input plays on $V \times \{orig\}$, moving from $(v,orig)$ to $(v,\sigma_{in})$ for $\sigma_{in} \in \Sigmain$, and player output plays on $V \times \Sigmain$, moving from $(v,\sigma_{in})$ to $(\delta_\mathcal{D}(v,\sigma_{in},\sigma_{out}),orig)$ for $\sigma_{out} \in \Sigmaout$. The colors are according to $v \in V$. This game is definable in the copying $\mathcal{M} \times (|\Sigmain|+1)$.

A causal operator implementing $R_\mathcal{D}$ exists if and only if $(v_0,orig)$ is in the winning region of player output in the game $G$. This is a statement in $MTh\big(\mathcal{M} \times (|\Sigma_{in}|+1)\big)$, because the winning regions of the parity game $G$ are definable in $G$, and $G$ is definable in $\mathcal{M} \times (|\Sigmain|+1|)$. By Corollary \ref{copyingRecursive}, $MTh\big(\mathcal{M} \times (|\Sigma_{in}|+1)\big)$ is recursive in $MTh(\mathcal{M})$.

\item Since $\mathcal{M}$ has the selection property, so does $\mathcal{M}'$. The game $G$ is definable in $\mathcal{M}'$ and has bounded out degree. To apply Theorem \ref{boundedOutDegree}, a local linear order of $G$ must be defined. To do so, first take any linear order on the finite set $\Sigmain \cup \Sigmaout$. From this order, a local linear order of $G$ may easily be defined.

It follows that there are $\mathcal{M}'$-definable uniform memoryless winning strategies for both players. In particular, if player output wins, there is a formula $\phi(x,y)$ defining in $\mathcal{M}'$ a winning strategy for player output.

By Lemma \ref{copyingLem}, $\phi(x,y)$ may be rewritten as a tuple of $\MSO$-formulas \[\langle \tilde{\phi}_{\sigma_{in}}(v,v')\ |\ \sigma_{in} \in \Sigmain\rangle\] such that for every $v,v' \in M$ \[\mathcal{M}' \models \phi((v,\sigma_{in}),(v',Orig)) \iff \mathcal{M} \models \tilde{\phi}_{\sigma_{in}}(v,v').\]

As to the transducer, its vertices and initial vertex are the vertices and initial vertex of the DA. The transition and output function $\delta_\mathcal{T}$ uses the linear order we've defined on $\Sigmaout$: \[\big(\delta_\mathcal{T}(v,\sigma_{in})=(v',\sigma_{out})\big) \iff \]\[\tilde{\phi}_{\sigma_{in}}(v,v') \wedge \big(\sigma_{out}\ minimal\ s.t.\ \delta_\mathcal{D}(v,\sigma_{in},\sigma_{out})=v')\big).\]

\item Follows from the proof of (2).
\end{enumerate}
\end{proof}

The above result guarantees the existence of a definable implementing transducer over a finite alphabet whenever the structure has the selection property, and in particular in $(\alpha,<,\bp)$ for $\alpha < \omega^\omega$. In the parameterless case:
\begin{cor}
Let $\Sigmain, \Sigmaout$ be finite alphabets, and let $\mathcal{M}$ be one of the following structures:
\begin{enumerate}
\item $(\alpha,<)$ for $\alpha < \omega^\omega$.
\item $\mathcal{T}_k$.
\end{enumerate}

Let $\mathcal{D}$ be an $\mathcal{M}$-definable deterministic acceptor over $\Sigmain \times \Sigmaout$.
\begin{enumerate}
\item (Decidability) It is decidable whether there exists a causal operator that implements $R_\mathcal{D}$.
\item (Definability, Computability) If there is a causal operator that implements $R_\mathcal{D}$, then there is an $\mathcal{M}$-definable transducer over $\Sigmain$ and $\Sigmaout$ that implements $R_\mathcal{D}$.

Moreover, it is computable to find a definition of such a transducer.
\end{enumerate}

\end{cor}
Note that the out-degree of $\mathcal{D}$ is bounded by $|\Sigmain\times \Sigmaout|$.

\subsection{Generalized Church Problem over infinite alphabets} \label{GCPInfSect}

In the case of finite alphabets, the input-output relation was defined in terms of an $\mathcal{M}$-definable transition system. The first step of the solution of the generalized Church problem over finite alphabets then involved the definition of a parity game. That parity game represented the transition system in the sense that player output wins the game if and only if the input-output relation defined by the transition system may be implemented by a causal operator.

For the proof to hold, this parity game had to be defined in a copying of the structure $\mathcal{M}$. Indeed, the game was defined in $\mathcal{M} \times (|\Sigmain|+1)$,  and each position of player output was composed of the previous position and the move player input has just played.

This cannot be done whenever $\Sigmain$ is infinite. For this reason, in the generalized Church
problem over infinite alphabets one cannot, in general, interpret the game arena in a deterministic automaton, or in a structure in which the automaton itself is interpreted. Instead, the input--output relation is specified directly by an edge-labeled game arena.

The framework below captures $\mathbb{N}$-$\MSO$ automata \cite{BRUTSCHLT17, BrutschT22} -- see also Remark \ref{RemNAutomata} and Theorem \ref{NAutThm}.

\begin{defi}
\begin{enumerate}
\item An \textbf{edge labelled parity game} $\mathcal{G}$ over $\Sigmain$ and $\Sigmaout$ is \[(V_{in},V_{out},\Sigmain,\Sigmaout,v_0,\delta_\mathcal{G},\Pi_0,\dots,\Pi_{d-1})\] where $V_{in},V_{out}$ are sets of nodes, $\Sigmain, \Sigmaout$ are alphabets, $v_0 \in V_{in}$ is an initial node, $\Pi_0,\dots,\Pi_{d-1}$ are colors of $V$ and $\delta_\mathcal{G}$
is a transition function
\[\delta_\mathcal{G} : (V_{in} \times \Sigmain) \cup (V_{out} \times \Sigmaout) \to V\] that maps $(V_{in} \times \Sigmain)$ to $V_{out}$ and $(V_{out} \times \Sigmaout)$ to $V_{in}$.
\item Given an edge labelled parity game $\mathcal{G}$, we denote by $\tilde{\mathcal{G}}$ the (non edge labelled) parity game constructed from $\mathcal{G}$ by forgetting the labels of the edges, which is, $\tilde{\mathcal{G}}=(V_{in},V_{out},E,\Pi_0,\dots,\Pi_{d-1})$, where for $x \in V_{in}, y \in V_{out}$ \[E(x,y) \iff \exists \sigma_{in} \in \Sigmain\ (\delta_\mathcal{G}(x,\sigma_{in})=y)\] and similarly for $x \in V_{out},\ y \in V_{in}$.
\item Given an edge labelled parity game $\mathcal{G}$, a \textbf{play}, a \textbf{strategy}, and \textbf{winning} plays and strategies, are the ones of the parity game $\tilde{\mathcal{G}}$.
\item Let $x,y$ be $\omega$-words over $\Sigmain$ and $\Sigmaout$, respectively. The \textbf{play induced by} $(x,y)$ is a play $(v_0,w_0,v_1,w_1,\dots)$ such that $v_0$ is the initial node, and for every $n$ \[w_n=\delta_\mathcal{G}(v_n,x_n),\] \[v_{n+1}=\delta_\mathcal{G}(w_n,y_n).\] We will say that $(x,y)$ are \textbf{winning for player output} if their induced play is winning for player output.
\item The \textbf{input-output relation defined by} $\mathcal{G}$ is the set of pairs $(x,y) \in {\Sigmain}^\omega \times {\Sigmaout}^\omega$ that are winning for player output. It is denoted by $R_\mathcal{G}$.

\end{enumerate}
\end{defi}
Note that the transition relation of an edge labelled parity game is deterministic and complete.

\begin{defi}
Let $\mathcal{M}$ be a structure. An edge labelled parity game $\mathcal{G}$  over $\Sigmain$ and $\Sigmaout$ is \textbf{definable in} $\mathcal{M}$ if:
\begin{itemize}
\item $V_{in}$, $V_{out}$, $\Sigmain$, $\Sigmaout$, $v_0$, $\Pi_0,\dots,\Pi_{d-1}$ are definable subsets of $\mathcal{M}$.
\item The function $\delta_{\mathcal{G}}$ is definable in $\mathcal{M}$
\end{itemize}
\end{defi}

\begin{prob}[Generalized Church Problem for $\mathcal{M}$-definable parity games]
Let $\mathcal{G}$ be an $\mathcal{M}$-definable edge labelled parity game over $\Sigmain$ and $\Sigmaout$.

Is there a causal operator $f:{\Sigma_{in}}^\omega \to {\Sigma_{out}}^\omega$ that implements $R_\mathcal{G}$?
\end{prob}

When the alphabet becomes infinite, transforming a winning strategy into an implementing transducer carries a new challenge: even when the next game position is known, there might be infinitely many letters leading to that same game position. To overcome it, we introduce the following definitions:
\begin{defi}
\begin{enumerate}
\item We say that $\mathcal{M}$ has \textbf{first order uniformization} if every $\MSO(\mathcal{M})$-formula $\phi(\bx,\by)$ with first order \textbf{free variables} has a uniformizer with domain variables $\bx$ over $\mathcal{M}$.
\item We say that $\mathcal{M}$ has \textbf{weak first order uniformization} if every $\MSO(\mathcal{M})$-formula $\phi(x,y)$ with first order \textbf{free variables} has a uniformizer with domain variable $x$ over $\mathcal{M}$.
(Here there is only one domain variable and one image variable.)
\end{enumerate}
\end{defi}
Note that if $\mathcal{M}$ has (weak) first order uniformization, then so do the copyings of $\mathcal{M}$.

\smallskip

\smallskip
\begin{thm} \label{gcpInfinite}
Let $\mathcal{M}$ be a structure, and let $\mathcal{G}$ be an $\mathcal{M}$-definable edge labelled parity game over $\Sigmain$ and $\Sigmaout$.

\begin{enumerate}

\item (Decidability) It is recursive in $MTh(\mathcal{M})$ whether there exists a causal operator that implements $R_\mathcal{G}$.
\end{enumerate}
If in addition:
\begin{enumerate}[(a)]
\item$\mathcal{M}$-definable parity games have $\mathcal{M}$-definable uniform memoryless winning strategies,
\item $\mathcal{M}$ has weak first order uniformization,
\end{enumerate}
then:
\begin{enumerate}
\setcounter{enumi}{1}
\item (Definability) If there is a causal operator that implements $R_\mathcal{G}$, then there is an $\mathcal{M}$-definable transducer over $\Sigmain$ and $\Sigmaout$ that implements $R_\mathcal{G}$.

Moreover, the vertices of the transducer are the game vertices of player input, and if $\delta_\mathcal{T}(v,\sigma_{in})=(v',\sigma_{out})$ then $\delta_\mathcal{G}(\delta_\mathcal{G}(v,\sigma_{in}),\sigma_{out})=v'$.
\item (Computability) If there is a causal operator that implements $R_\mathcal{G}$, $\mathcal{M}$-definable parity games have \textbf{computable} $\mathcal{M}$-definable uniform memoryless winning strategies, and weak first order uniformization is solvable, then it is computable to find a definition of a transducer that implements $R_\mathcal{G}$.
\end{enumerate}
\end{thm}

\begin{proof}
The proof follows similar steps as the ones taken in the finite alphabet case. In the first step the edge labels are forgotten in order to form a standard parity game. In the second step, a definable memoryless winning strategy is found. In the third step, this strategy is transformed into an implementing transducer. On each move of player output, the transducer will use weak first order uniformization to pick up a letter $\sigma_{out}\in\Sigmaout$ that is compatible with the defined strategy.

\begin{enumerate}
\item Let $\tilde{\mathcal{G}}$ be the (non edge labelled) parity game constructed from $\mathcal{G}$ by forgetting the labels of the edges. Clearly, player output wins the game $\tilde{\mathcal{G}}$ if and only there exists a causal operator that implements $R_\mathcal{G}$. The game $\tilde{\mathcal{G}}$ is $\mathcal{M}$-definable, hence deciding whether $v_0$ is in the winning region of player output is a sentence in $MTh(\mathcal{M})$.

\item To define the transducer, we first let its vertices be the $\mathcal{G}$-vertices of player input, and its initial vertex be the same as the one of $\mathcal{G}$.

Let $\phi(v_{out},v_{in})$ be a definition of a uniform memoryless winning strategy for player output. Let \[\tilde{\phi}(v_{out},\sigma_{out}) =\]\[ \exists v_{in}\ \Big(\phi(v_{out},v_{in}) \land \big(\delta_\mathcal{G}(v_{out},\sigma_{out})=v_{in}\big)\Big),\] and let $\tilde{\phi}^*(v_{out},\sigma_{out})$ be a uniformizer of $\tilde{\phi}$ with domain variable $v_{out}$ over $\mathcal{M}$.

The transition and output function $\delta_\mathcal{T}$ is then defined as follows: \[\delta_\mathcal{T}(v_{in},\sigma_{in}) = ({v'_{in}},\sigma_{out}) \iff \]  \[\exists v_{out}\  \Big((\delta_\mathcal{G}(v_{in},\sigma_{in})=v_{out}) \wedge \phi(v_{out},{v'_{in}}) \wedge \]\[ \tilde{\phi}^*(v_{out},\sigma_{out})\Big).\]
\item Follows from the proof of (2).
\end{enumerate}
\end{proof}

As a trivial remark, all ordinals have first order uniformization. The following theorem implies the first order uniformization of $\mathcal{T}_k$, and also of $(\mathcal{T}_k,\bar{P^*})$ from Example \ref{YesSelExample} (see \cite{Semenov} p.172, in which this result for $\mathcal{T}_k$ is attributed to Muchnik):

\begin{restatable}[Uniformization of relations with first order domain variables]{thm}{unifoStatement}\label{unifoT2} Let $\bp \subseteq T_k$ be such that the class of subtrees $(T_{\geq z},<,\bE,\bp)$ has the (solvable) selection property. Let $\phi(\bx,\bY,\bp)$ be a definable relation, where $\bx$ are first order variables. Then there is (a computable) $\psi(\bx,\bY,\bp)$ that uniformizes $\phi$ with domain variables $\bx$ over the structure $(\mathcal{T}_k,\bp)$.
\end{restatable}

\begin{proof}

By Theorem \ref{AutLogEqTrees}, $\phi$ is equivalent to a $k$-tree automaton. Let $\mathcal{A}$ be a $k$-tree automaton over $2^{\{1,\dots,n\}}$, $n=|\bx|+|\bY|+|\bp|$, that is equivalent to $\phi$.

For every $q\in Q_\mathcal{A}$, let $\phi_q(\bY,\bp)$ say that assuming $\bx=\emptyset$, the tree has an accepting run from $q$. For every $q$, let $\phi_q^*(\bY,\bp)$ select $\phi_q$ over the class of subtrees of $(
\mathcal{T}_k,\bp)$.

For every $q$, add predicates $q^0,\dots,q^{k-1}$. These predicates are for coding a run of the automaton, in the sense that $q^i(x)$ represents a run in which the $i$'th successor of $x$ is in state $q$. Let $\bar{q}$ stand for all of these predicates, and $B(\bx)=B(x_0,\dots,x_{m-1})$ stand for the union of paths $[\epsilon,x_0]\cup\dots\cup[\epsilon,x_{m-1}]$.

Let $\eta(\bx,\bY,\bp,\bar{q})$ say that
\begin{itemize}
\item $\bar{q}$ code a legal run of $\mathcal{A}$ on the union of paths $B(\bx)$.
\item If $y\in B(\bx)$, $q^i(y)$, $z$ is the $i$'th successor of $y$, and $z$ is not in $B(\bx)$, then \[T_{\geq z} \models \exists \bY\ \phi_q^*(\bY,\bp).\]
\end{itemize}

The lexicographic order on $T_k$ is a definable linear order, and on the finite set $B(\bx)$ it is a definable well-order. When a set is well-ordered by a definable relation, there is a definable linear order of its subsets given by: $A_0\prec A_1$ if the first element of $A_0 \triangle A_1$ is an element of $A_0$. It follows that there is a definable well-order $\prec$ of the subsets of $B(\bx)$. Let $\eta^*(\bx,\bY,\bp,\bar{q})$ state that $\eta$ holds, and $\bar{q}, \bY\restriction_{B(\bx)}$ are minimal with respect to $\prec$.

The uniformizer $\psi(\bx,\bY,\bp)$ then states that there is $\bar{q}$ such that
\begin{itemize}
\item $\eta^*(\bx,\bY,\bp,\bar{q})$, and
\item If $y\in B(\bx)$, $q^i(y)$, and the $i$'th successor of $y$, $z$, is not in $B(\bx)$, then \[T_{\geq z} \models \phi_q^*(\bY,\bp).\]
\end{itemize}

\end{proof}

\begin{cor} \label{GcpOrdTree}
Let $\mathcal{M}$ be one of the following structures:
\begin{enumerate}
\item $(\alpha,<)$ for $\alpha < \omega^\omega$.
\item $\mathcal{T}_k$.
\item $(\alpha,<,\bp)$ for $\alpha < \omega^\omega$.
\item $(\mathcal{T}_k,\bp)$ such that the class of subtrees  \[(T_{\geq z},<,\bE,\bp)\] has the selection property.
\end{enumerate}

Let $\mathcal{G}$ be an $\mathcal{M}$-definable edge labelled parity game over $\Sigmain$ and $\Sigmaout$.

\begin{enumerate}
\item (Decidability) In cases (1) and (2), it is decidable whether there exists a causal operator that implements $R_\mathcal{G}$.

\item (Definability) If there is a causal operator that implements $R_\mathcal{G}$, then there is an $\mathcal{M}$-definable transducer over $\Sigmain$ and $\Sigmaout$ that implements $R_\mathcal{G}$.
\item (Computability) In cases (1) and (2), it is computable to find a definition of such a transducer.
\end{enumerate}
\end{cor}

\begin{rem} \label{RemNAutomata}
Brütsch and Thomas \cite{BrutschT22} gave a solution to the generalized Church problem for the case of $\mathbb{N}$-$\MSO$ automata.

An $\mathbb{N}$-$\MSO$ \textbf{automaton} $\mathcal{A}$ consists of $\langle Q,q_0,\phi_{p,q}, c \rangle$ where $Q$ is a finite set of states, 
$q_0 \in Q$ is an initial state, $\langle \phi_{p,q}(m,k,n)\ |\ p,q \in Q \rangle$ are $\MSO$ formulas over $\langle \omega,<\rangle$, and $c$ is a coloring function of $Q$.
A \textbf{configuration} is a pair $(m,q)$ where $m \in \omega$, $q \in Q$.
A \textbf{run} of an $\mathbb{N}$-$\MSO$ automaton $\mathcal{A}$ over a word $w \in \mathbb{N}^\omega$ is an $\omega$-sequence of configurations, such that the first configuration is $(0,q_0)$, and if the $i'th$ configuration is $(m_i,q_i)$ and the $i+1'th$ configuration is $(m_{i+1},q_{i+1})$ then $\omega \models \phi_{q_i,q_{i+1}}(m_i,w(i),m_{i+1}).$
An $\mathbb{N}$-$\MSO$ automaton is \textbf{deterministic} if the formulas $\phi_{p,q}(m,k,n)$ are functional in $p,m,k$, which is, for every $p,m,k$ there is a unique pair $q,n$ such that $\omega \models \phi_{p,q}(m,k,n)$.

\begin{thm}\label{NAutThm} \cite{BrutschT22}
Let $\mathcal{A}$ be a deterministic $\mathbb{N}$-$\MSO$ automaton, and let $L$ be the language of words accepted by $\mathcal{A}$. Let $\mathcal{G}(L)$ be the two player game with $\omega$ moves in which each player plays a natural number, and player $I$ wins if the result of the game is in $L$.
\begin{enumerate}
\item (Determinacy) One of the players has a winning strategy in $\mathcal{G}(L)$.
\item (Decidability) Finding the winner of $\mathcal{G}(L)$ is recursive.
\item (Definability) There is an $\mathbb{N}$-$\MSO$ transducer that is a winning strategy for the winner.
\item (Computability) It is computable to find an $\mathbb{N}$-$\MSO$ transducer that is a winning strategy for the winner.
\end{enumerate}
\end{thm}

This theorem easily follows from our results using the following steps:

\begin{enumerate}
\item Present the game as an edge labelled parity game definable in $\omega \times Q \times \{1,2\}$.

The nodes are triples  $(i,q ,t)$ where $i\in \om$, $q\in Q$ is a state, and $t\in \{1,2\}$ indicates that Player $t$ moves from this node.
Each transition changes the third coordinate, and there is a transition from $(m,p,t)$ to $(n,q,3-t)$ labelled by $k$
if $\phi_{p,q}(m, k,n)$ holds.
\item By Corollary \ref{GcpOrdTree}, we can find (compute) a transducer that is definable in $\omega \times Q \times \{1,2\}$ and induces a winning strategy for the winner. (By standard arguments, Corollary \ref{GcpOrdTree} may be adapted to apply to a copying of $\omega$.)

\item It is left to (effectively) transform the transducer into an $\mathbb{N}$-$\MSO$ transducer, and that is done by standard arguments.
\end{enumerate}

\end{rem}
\section{Conclusion}

The aim of our paper is to present a first systematic study of the generalized Church synthesis problem, a problem that hasn't been stated before. The original Church synthesis problem is about finite automata, whereas its generalized version is an almost identical synthesis problem for definable infinite state systems. Apart from being a natural generalization, it can also have future applications, given that finite state systems usually serve as models for hardware, whereas software is frequently represented by infinite state systems.

Although the Church synthesis problem has always been considered a key result in the area of logic, automata and games, there were very few attempts to expand it to a more general context, as we do here. For that purpose, we provided:
\begin{itemize}
\setlength\itemsep{0.001em}
\item A conceptual framework.
\item Some sufficient conditions.
\item Non-trivial technical results.
\end{itemize}
The outcome is a solution of the generalized Church problem in some important instances.

\section*{Acknowledgement}

We would like to thank Wolfgang Thomas for useful comments and discussions on earlier versions of this paper. We also thank the anonymous reviewers for their valuable comments and helpful suggestions, which improved the manuscript. In particular, we are grateful to an anonymous reviewer for drawing our attention to \cite{BraudCarayol} and \cite{CarayolThesis}.

\end{document}